\documentclass[a4paper,noarxiv,onecolumn, accepted=2022-02-06]{quantumarticle}

\usepackage{preamble_lat}
\newcommand{\jc}[1]{{\color{red}[JC: #1]}}
\newcommand{\fa}[1]{{\color{purple}[FA: #1]}}
\newcommand{\je}[1]{{\color{cyan}[JE: #1]}}
\newcommand{\jen}[1]{{\color{cyan}#1}}
\renewcommand{\jc}[1]{}
\renewcommand{\fa}[1]{}
\renewcommand{\je}[1]{}
\renewcommand{\jen}[1]{}

\begin{document}


\title{Gottesman-Kitaev-Preskill codes: A lattice perspective}

\author{Jonathan Conrad}\orcid{0000-0001-6120-9930}
\email{j.conrad1005@gmail.com}
\affiliation{Dahlem Center for Complex Quantum Systems, Physics Department, Freie
Universit{\"a}t Berlin, Arnimallee 14, 14195 Berlin, Germany}
\affiliation{Helmholtz-Zentrum Berlin f{\"u}r Materialien und Energie, Hahn-Meitner-Platz 1, 14109
Berlin, Germany}
\author{Jens Eisert}\orcid{0000-0003-3033-1292}
\affiliation{Dahlem Center for Complex Quantum Systems, Physics Department, Freie
Universit{\"a}t Berlin, Arnimallee 14, 14195 Berlin, Germany}
\affiliation{Helmholtz-Zentrum Berlin f{\"u}r Materialien und Energie, Hahn-Meitner-Platz 1, 14109
Berlin, Germany}
\author{Francesco Arzani}\orcid{0000-0002-4439-6962}
\email{frarzani@zedat.fu-berlin.de}
\affiliation{Dahlem Center for Complex Quantum Systems, Physics Department, Freie
Universit{\"a}t Berlin, Arnimallee 14, 14195 Berlin, Germany}

\date{\today}

\begin{abstract}

We examine general Gottesman-Kitaev-Preskill (GKP) codes for continuous-variable quantum error correction,  including concatenated GKP codes,  through the lens of lattice theory,  in order to better understand the structure of this class of stabilizer codes.  We derive formal bounds on code parameters,  show how different decoding strategies are precisely related,  propose new ways to obtain GKP codes by means of glued lattices and the tensor product of lattices and point to natural resource savings that have remained hidden in recent approaches.  We present general results that we illustrate through examples taken from different classes of codes,  including scaled self-dual GKP codes and the concatenated surface-GKP code.

\end{abstract}

\maketitle

\section{Introduction}

It is known that quantum physics can be harnessed to speed-up certain computational tasks.  However, it seems challenging to take advantage of such potential in practice.  Arguably the biggest hindrance to achieving reliable quantum information processing is the difficulty to control systems and sufficiently isolate them from their environment.  The field of \emph{quantum error correction (QEC)} has been developed to cope with errors by choosing an appropriate encoding to constrain the action of environmental noise on the stored information and by repeated measurements of an associated stabilizer group to remove entropy from the system.
  With advancing experimental control over \emph{superconducting} and \emph{integrated optical photonic} circuits,  recent years have witnessed growing interest in bosonic quantum error correction, where information is embedded  into bosonic or continuous-variable (CV) systems exploiting the redundancy available in the associate infinite dimensional Hilbert space.  In particular,  so-called \emph{Gottesman-Kitaev-Preskill (GKP)} codes \cite{GKP} have been increasingly studied as promising components towards building a scalable quantum computer \cite{Terhal_2020,  Grimsmo_Puri,  Bourassa_2021,  bartolucci2021fusionbased, noh2021low}. 
  
While the original proposal for the GKP code was made as early as in 2001,  only in recent years experimental efforts of realizing such codes have gained momentum.  Quantum error correction with single mode versions of the GKP code has recently been demonstrated experimentally in ion-traps \cite{Fluehmann_2019} and superconducting hardware \cite{Campagne_Ibarcq_2020} while a growing number of theoretical proposals investigate the possibility to concatenate such codes with known qubit codes \cite{fukui2017analog,  toricGKP,  surfGKP,  Haenggli_2020,Bourassa_2021} to achieve exponential error suppression by increasing the number of modes.  Most works in this direction focus on understanding the single mode GKP code as an effective encoded qubit,  with an effective noise channel,  and proceed in their analysis using the well established toolbox of qubit  quantum error correction.  While it allowed rapid progress,  this approach significantly limits the class of codes that can be considered and the analytical grasp on the continuous nature of the underlying system.  In general,  both the structure of the state and that of the stabilizer group of GKP codes can be understood in terms of lattices in phase-space,  which is often used to introduce and visualize single-mode examples.  But the deep connection with lattices also holds in the multi-mode case,  including but,  crucially,  not limited to codes obtained from concatenating a single-mode GKP code with qubit stabilizer codes.  The description of  (multi-mode) GKP codes in terms of symplectic lattices was already formulated in the seminal work \cite{GKP}.  However,  only a handful of examples can be found in the literature \cite{ Haenggli_2020,  Harrington_Thesis,  HarringtonPreskill,  oscillator_to_oscillator} that attempt to use it for a better understanding of the structure and properties of such codes,  for constructing new codes or for devising decoding algorithms.   

In this work we present an introduction to lattice theory for quantum error correction practitioners and show how the tools it provides can be fruitfully put to use for the analysis and design of GKP codes.  The upshot is that,  since lattices have long been studied in many branches of pure and applied mathematics,  from number theory to string theory,  they provide us with a rich toolbox whose potential remains largely untapped for quantum error correction.  By elaborating on the connection with lattices,  we move some further steps in this direction.  We hope that this will motivate a wider adoption of lattice techniques in the study of GKP- and general quantum error correcting codes.  The results in the present work suggest that this could lead to further insight as well as potential practical advantages towards fault-tolerant quantum computation with GKP codes.

To demonstrate the power of the lattice formalism, we derive a number of new results relating the parameters of GKP codes,  such as  the distance and encoding rate.  We study the equivalence of codes under symplectic transformations,  formulate the general decoding problem in terms of lattice quantities and provide two constructions leading to new codes that cannot be described (to the best of our knowledge) within the usual concatenated GKP-qubit paradigm.  This highlights that the latter only captures very special cases,  that,  while easy to conceive,  have no \emph{a priori} reason to be considered optimal in taking advantage of the underlying bosonic nature of the system.

This article is structured as follows. 
In Section \ref{sec:preliminaries},  we set the stage by introducing the relevant notions from lattice theory that will be used in the rest of this manuscript and we discuss their relation to GKP codes.
Section \ref{sec:bases} explores how the choice of the lattice basis corresponding to the stabilizer generators influences the properties of the latter.  We focus in particular on suitable generalizations of 
\emph{Calderbank-Shor-Steane (CSS)} codes and \emph{low-density-parity-check (LDPC)} codes to the GKP setting.  We show that significant savings in the number of stabilizer measurements with respect to the procedure usually employed for concatenated GKP-qubit codes are possible by measuring a minimal set of stabilizers at each error correction cycle.  Furthermore,  we show that concatenation of qubit stabilizer codes with single mode GKP codes provides a natural embedding of the qubit code into a lattice through the so-called \emph{Construction A} \cite{ConwaySloane}, 
which leads to a no-go result that impacts the complexity of the decoding problem,  in particular exact \emph{minimum-energy-decoding}.

In Section \ref{sec:symplectic},  we prove a necessary and sufficient condition for two GKP codes to be equivalent up to a symplectic transformation,  extending the single mode case mentioned in ref.~\cite{Haenggli_2020}.  This can be interpreted as code-switching by means of a Gaussian unitary.
Section \ref{sec:distance} is focused on defining a notion of \emph{distance} for GKP codes via the length of the shortest non-trivial dual coset element of the lattice that describes the stabilizers.  This definition is motivated by the structure of common noise models used for GKP codes in the literature.  We discuss tradeoffs between the distance and other code parameters,  such as encoding rate and the ``length'' of stabilizers.  Furthermore,  we show how to estimate the distance of a GKP code from the theta functions of the corresponding lattice,  which are central quantities in lattice theory.  Finally,  we show that the distance of a CSS \emph{qubit} stabilizer code is fully determined by the weight distribution of its stabilizer group,  i.e.,  the number of stabilizers of particular Hamming weight in their representation by symplectic binary vectors.  Although a similar statement can be derived from the seminal work of Shor,  Laflamme \cite{ShorLaflamme} and Rains \cite{Rains},  it follows straightforwardly from the lattice theoretic concepts we introduce,  showcasing that a lattice theoretic standpoint can also benefit the examination of qubit stabilizer codes.

Theta functions are central in the general formulation of the maximum likelihood (coset) decoding problem for general GKP codes,  as we detail in Section \ref{sec:MLD}.  This hints at the existence of new decoding strategies that do not rely on the concatenated code structure,  and are hence more general than those adapted from qubit codes.  A detailed discussion of such strategies is beyond the scope of the present work and is therefore postponed to a separate publication.

Finally,  we discuss in Section \ref{sec:constructions} two constructions,  based on  glued lattices and the lattice tensor product operation,  respectively,  to obtain GKP codes beyond scaled and concatenated codes,  the two main code constructions featured in previous works and the rest of this work.  We foresee that such alternative constructions may become of practical relevance in the near future,  especially in setups with limited locality constraints,  such as photonic architectures.  Furthermore,  we show that some properties of these codes can be computed easily,  potentially making them also useful examples to probe the limits of GKP codes.  Concluding remarks and future perspectives in Section \ref{sec:conclusions} complete this work.

\section{Preliminaries and notation \label{sec:preliminaries}}

In this section, we review the basic setup of GKP codes as proposed in ref.~\cite{GKP} and introduce some basic concepts of lattice theory to analyze their structure.  Good introductions to lattice theory are found in refs.~\cite{CryptoLectureNotes,  ConwaySloane},  from which many of the following definitions and lattice theoretic arguments presented here are adapted.
A \emph{stabilizer code} \cite{gottesman_thesis} embeds information into a physical Hilbert space by defining the code-space as fixed point of a finitely generated,  abelian group $\mathcal{S}$,  i.e.,  logical information is stored in terms of state vectors 
\begin{align}
\ket{\psi}:\; s\ket{\psi}=+1 \ket{\psi}\,  \forall s \in \mathcal{S}.
\end{align}
The \emph{GKP code} \cite{GKP} is a stabilizer code acting on the Hilbert space of $n$ bosonic modes,  where stabilizers are given by displacement operators
\begin{align}
D\lr{\bs{\xi}} = \prod_k \exp \left\{-i\sqrt{2\pi}u_k \hat{p}_k + i \sqrt{2\pi} v_k \hat{q}_k \right\} = \exp\left\{-i \sqrt{2\pi} \bs{\xi}^\mathrm{T} J \hat{\bs{x}}\right\},\quad \bs{\xi} = \lr{\ba{c} \bs{u} \\ \bs{v} \ea} ,\quad  \bs{\hat{x}} =  \lr{\ba{c} \hat{\bs{q}} \\ \bs{\hat{p}} \ea}, \label{eq:displConvJ}
\end{align}
where \begin{align}
J=J_{2n}=\begin{pmatrix}0 & 1 \\ -1 & 0\end{pmatrix} \otimes I_n= \begin{pmatrix}
0 &  I_n \\
-I_n  & 0 \end{pmatrix}
\end{align} 
is the \emph{symplectic form} and $\bs{\xi}\in\mathbb{R}^{2n}$ are phase-space vectors that specify the displacement.   When no index is specified, we follow the convention $J=J_{2n}$ and $I=I_{2n}$.  We also  adopt the convention $\hbar=1$,  such that $[\hat{x}_k,\hat{x}_l]=iJ_{k,l}$ for $k,l=1,\dots, 2n$. 
Throughout this work,  we use bold symbols to denote (column) vectors.
\emph{Displacement operators} implement a shift of the  $2n-$quadratures
\begin{equation}
D^{\dagger}(\bs{\xi}) \hat{\bs{x}}D(\bs{\xi})=  \hat{\bs{x}}+ \sqrt{2\pi}\bs{\xi},
\end{equation}
commute as 
\begin{equation}
D(\bs{\xi})D(\bs{\eta})=e^{-i 2\pi\bs{\xi}^TJ\bs{\eta}}D(\bs{\eta})D(\bs{\xi}), \label{eq:disp_commute}
\end{equation}
and satisfy
\begin{equation}
D(\bs{\xi})D(\bs{\eta})=e^{-i \pi\bs{\xi}^TJ\bs{\eta}}D\lr{\bs{\xi}+\bs{\eta}}. \label{eq:disp_close}
\end{equation}
Given the metaplectic representation of a  symplectic operator~\cite{pramana} $S\mapsto U_S,\,  S^TJS=J $,  we have
\begin{equation}
U_S  \hat{\bs{x}} U_S^{\dagger}=S \hat{\bs{x}},
\end{equation}
such that 
\begin{equation}\label{eq:symplDisplTrasf}
U_{S^{-1}}D(\bs{\xi})U_{S^{-1}}^\dagger=\exp\lr{-i\bs{\xi}^T J S^{-1}\hat{\bs{x}}}=\exp\lr{-i\bs{\xi}^T S^T J S S^{-1}\hat{\bs{x}}}=D(S\bs{\xi}).
\end{equation}
Displacement operators are orthogonal,  \begin{equation}
\Tr\lrq{D^{\dagger}\lr{\bs{\xi}} D(\bs{\eta})} =  \delta^{(2n)}\lr{\bs{\xi}-\bs{\eta}}, \label{eq:disp_orthogonal}
\end{equation} 
and form a complete basis such that any operator $\hat{O} \in \mathcal{L}\lr{\mathcal{H}}$ can be resolved as\footnote{To
be precise,  this must be restricted to trace-class operators and one should formulate exceptions for operators such as $D\lr{\beta}$ to satisfy eq.~\eqref{eq:disp_orthogonal}.} 
\begin{equation}
\hat{O}=\int d\bs{x} o(\bs{x}) D\lr{\bs{x}} = \int d\bs{x}  \Tr\lrq{D^{\dagger}\lr{\bs{x}} \hat{O}}  D\lr{\bs{x}}.
\end{equation}
We can now define the GKP stabilizer group.

\begin{mydef}[GKP stabilizer group \cite{GKP}]\label{def:GKP}
The stabilizer group of a GKP code is given by a set of displacements \jc{there are $\langle \cdot \rangle$ brackets down there,  hence ``given" by and not ``generated"}
\begin{equation}
\mathcal{S}:=\langle D\lr{\bs{\xi}_1},\dots,  D\lr{\bs{\xi}_{2n}} \rangle,
\end{equation}
where $\lrc{\bs{\xi}_i}_{i=1}^{2n}$ are linearly independent and we have  $\bs{\xi}_i^TJ\bs{\xi}_j \in \mathbb{Z} \forall\,
 i,\, j$.
\end{mydef}

Exploiting the structure of the displacement operators,  the GKP construction defines a 
stabilizer group isomorphic to a \emph{lattice} with generator matrix
\begin{equation}
M=\begin{pmatrix}
\bs{\xi}_1^T\\ \vdots \\ \bs{\xi}_{2n}^T
\end{pmatrix},
\end{equation}%
which is simply the set of integer linear combinations of basis elements%
\begin{equation}
\mathcal{L}=\mathcal{L}\lr{M}=\lrc{\bs{\xi} \in \mathbb{R}^{2n}\big\vert \;\bs{\xi}^T=\bs{a}^TM, \, \bs{a}\in \mathbb{Z}^{2n \times 2n}} .
\end{equation}%
We follow the convention of ref.~\cite{GKP} such that the basis vectors of the lattice $\mathcal{L}$ constitute the \textit{rows} of the generator matrix.

For $\mathcal{S}$ to constitute a stabilizer group,  it needs to be (1.)  a group and (2. ) abelian.  By construction in Definition  \ref{def:GKP},  $\mathcal{S}$ is a group and by means of 
 eq.~\eqref{eq:disp_commute} it can be observed that $\mathcal{S}$ is abelian if and only if the \textit{symplectic Gram matrix} associated with \textit{any} generator $M$ of the lattice%
\begin{equation}
A:=MJM^T
\end{equation}
has only \textit{integer} entries.  We then say that the lattice $\mathcal{L}(M)$ is \textit{symplectically integral}.  
 \je{Maybe this is redundant, but one could say that the lattice is uniquely defined by its generator $M$. But the converse is not true, as we can state already here, s there are many $M$ corresponding to the same lattice.} \jc{this discussion is further below}
 The symbol $A$ will be reserved to the symplectic Gram matrix throughout the manuscript and although it depends on the particular generator $M$ we will mostly leave this understood when the generator is clear from the context. 
 
Using eq.~\eqref{eq:disp_close},  it can be verified that each element of the stabilizer group $\mathcal{S}$ can be written as
\begin{align}
\prod_{i=1}^{2n} D\lr{\bs{\xi}_i}^{a_i} &= \prod_{i=1}^{2n} D\lr{a_i \bs{\xi}_i}  \\
&=e^{i\pi\bs{a}^T  A_{\lowertriangle}  \bs{a}} D\lr{(\bs{a}^TM)^T},
\end{align}
\jc{
More elaborate derivation: Let $\bs{a}_{[i]}$ denote the vector $\bs{a}=\bs{a}_{[2n]}$ with all entries $a_j=0\, \forall j>i$ and let $\phi_{i}$ denote the cumulative global phase when simplifying the product (row vector convention in $D\lr{}$ here)
\begin{align}
\prod_{i=1}^{2n} D\lr{\bs{\xi}_i}^{a_i} 
&= \prod_{i=1}^{2n} D\lr{a_i \bs{\xi}_i} \\
&=\lrq{\prod_{i=1}^{2n-1} D\lr{a_i \bs{\xi}_i} } D\lr{a_{2n} \bs{\xi}_{2n}} \\
&= e^{i\phi_{2n-1}} D\lr{\bs{a}_{[2n-1]}^TM} D\lr{a_{2n} \bs{\xi}_{2n}} \\
&= e^{i\phi_{2n-1}}  e^{-i\pi \bs{a}^T_{[2n-1]}MJ(a_{2n}\bs{\xi}_{2n})} D\lr{\bs{a}_{[2n]}^TM}.
\end{align}
We combined the displacements using eq.  \eqref{eq:disp_close},  from where we can already see that $e^{i \phi_j}=\pm 1 \forall j$ when $A$ is integer.  The expression above allows us to write down the recursion
\begin{align}
\phi_1&=0 \\
\phi_i 
&=\phi_{i-1}-\pi \bs{a}^T_{[i-1]}MJ(a_{i}\bs{\xi}_{i}) \\
&=\phi_{i-1}+\pi \sum_{j=1}^{i-1} a_i A_{i, j} a_j.
\end{align} 
Computing the recursion,  we observe that 
\begin{equation}
\phi_{2n}=\pi \sum_{i>j} a_i A_{i,j} a_j = \pi \bs{a}^T A_{\lowertriangle} \bs{a},
\end{equation}
where $A_{\lowertriangle}:\; (A_{\lowertriangle})_{i,j}=A_{i,j}\, \forall\, i>j$  is the lower triangular matrix of $A$.
}
where $\bs{a}\in \mathbb{Z}^{2n}$ is an integer vector and  $A_{\lowertriangle}:\; (A_{\lowertriangle})_{i,  j}=A_{i,  j}\, \forall\, i>j$  is the lower triangular matrix of $A$. 
The stabilizer group $\mathcal{S}$ is hence given by
\begin{equation}
\mathcal{S}=\lrc{e^{i\phi_M\lr{\bs{\xi}}} D\lr{\bs{\xi}} \vert \bs{\xi} \in \CL},\label{eq:GKP_2}
\end{equation}
where
\begin{equation}
\phi_M\lr{\bs{\xi}}=\pi \bs{a}^T A_{\lowertriangle} \bs{a}, \; \bs{a}^T=\bs{\xi}^TM^{-1}
\end{equation}
denotes the phase-sector when we have chosen $M$ to be the \emph{pivot basis},  i.e.  the set of basis vectors for which  each associated displacement operator is fixed to eigenvalue $+1$ by definition \ref{def:GKP}.  We will also use the notation $\mathcal{S}=\mathcal{S}(M, \phi_M)$  to specify GKP stabilizer groups from its generator $M$ when the phase-sector is relevant.  The pair $(M, \phi_N)$  specifies that the generator for the stabilizer group $D\lr{M_i^T}$ is fixed by eigenvalue $e^{i\phi_N\lr{M_i^T}} $ in code space.  Consistently,  we have $(M, \phi_M)=(M, 0)$.

When $\CL$ is symplectically integral,  we also have $e^{i\phi\lr{\bs{\xi}}}=\pm 1\,  \forall \, \bs{\xi} \in \CL$,  which is also already apparent from eq.~\eqref{eq:disp_close}. The additional phases in front of the displacement operators act trivially on the projective Hilbert space,  such that they are irrelevant when considering stabilizer operators $\mathcal{S}$ as \textit{symmetries} of the effective code.  However,  they play a role in error correction,  where the goal is to fix each stabilizer $s\in \mathcal{S}$ to the $+1$ eigenvalue to define a code space and the phases determine to which $\pm 1$ eigenvalue the associated displacement operator needs to be fixed.

\begin{remark}[Definition of the GKP stabilizer group]\normalfont
We remark that the definition of the GKP stabilizer group we make here differs from the one presented in ref.~\cite{GKP},  which has been stated as $\mathcal{S}'=\lrc{D\lr{\bs{\xi}} \vert  \bs{\xi} \in \CL}$,  where $\CL$ is a symplectically integral lattice.  Although this defines a commuting set of stabilizers,  $\mathcal{S}'$ is in fact not closed in general due to the extra phases appearing when combining two displacement operators eq.~\eqref{eq:disp_close} and would also not equal $\mathcal{S}$ as defined above.  Definition \ref{def:GKP} follows more closely the practical reality,  where error correction is performed by measuring  minimal generating set of the stabilizer group and by fixing (correcting) their eigenvalues to $+1$. 

The definition $\mathcal{S}'$ is closed and identical to $\mathcal{S}$ in the special case where $A \in \lr{2\mathbb{Z}}^{2n \times 2n}$,  i.e.,  the associated lattice $\CL$ is \emph{symplectically even}.  We found,  however, that this would restrict us too severely in the classes of codes to analyze (e.g.  the concatenated GKP - surface code introduced later is described by a symplectically integral,  but not symplectically even lattice). \je{Really, is the latter obvious?}\jc{not really super obvious,  but I don't think its worth spelling out here...}
We illustrate the difference between the sets $\mathcal{S}$ and $\mathcal{S}'$ using the stabilizer group for the sensor state \cite{displacement_sensor},
\begin{equation}
\mathcal{S}_{\text{sensor}}=  \langle D\lr{\bs{e}_1},\, D\lr{\bs{e}_2}  \rangle= \lrc{(-1)^{mn }D\lr{m\bs{e}_1+n\bs{e}_2},\; (m,n)\in \mathbb{Z}^2 },
\end{equation} where $\bs{e}_j$ are canonical basis vectors,  so that code-states should have eigenvalue $-1$ on displacement for which $mn$ is odd if they have eigenvalue $+1$ on each generator $D\lr{\bs{e}_1},\, D\lr{\bs{e}_2}$.   Given the conventional definition $\mathcal{S}'$,  $D\lr{\bs{e}_1}D\lr{\bs{e}_2}=(-1)D\lr{\bs{e}_1+\bs{e}_2}$ would formally not be included in the stabilizer group. \\
\end{remark}

A system of $n$ bosonic modes can be associated with an infinite dimensional Hilbert space and a $2n-$dimensional phase space,  whose axes correspond to the quadratures of the system,  each corresponding to an observable taking continuous values.  To encode discrete quantum information,  such as a qubit,  the state space needs to be fully ``discretized'' by introducing suitable constraints.  This is done by choosing a generating set for the stabilizer group with $2n$ linearly independent generators.  Each linearly independent generator can be seen as ``discretizing'' one direction in phase space.  This is the reason why  we require the lattice $\mathcal{L}$ to be \textit{full rank},  that is,  we demand that $M$ has full row-rank.  It is always possible to specify a lattice in $\mathbb{R}^{2n}$ using more than $2n$ basis vectors,  (we will see an example of such a scenario later),  but they can always be reduced to $2n$ linearly independent vectors,  a process for which a number of (efficient) algorithms are known \cite{CryptoLectureNotes}.  We note that also non-full rank (or \textit{degenerate}) lattices had recently been explored to define effective GKP codes \cite{Noh_2020,  oscillator_to_oscillator}.  The code space of such codes retains undiscretized quadratures which allows to encode and perform error correction on general CV states.  In this work we focus on encoding qubits or qudits defined via full-rank lattices while encoding of continuous information using degenerate lattices is briefly explained in appendix \ref{appendix:degenerate}.

Given a lattice $\mathcal{L}$,  the lattice basis is not unique.  Rather,  two generator matrices $M, \,M' $ generate the same lattice $\mathcal{L}\lr{M}=\mathcal{L}\lr{M'}$ if and only if there exists a unimodular matrix $U\in GL\lr{2n,\, \mathbb{Z}},\, |\det(U)|=1$ such that
\begin{equation}
M'=UM.
\end{equation}
Such transformation also transforms 
\begin{equation}
A \mapsto A'=UAU^T.
\end{equation}
Note that $A'$ has even entries if and only if $A$ does.  
Due to the phases appearing in eq.~\eqref{eq:GKP_2},  when a different basis $M'$ is used to fix the stabilizer group as in Definition  \ref{def:GKP},  the generating set for the stabilizer group needs to be chosen as 
\begin{equation}
\lr{M', \,  \phi_M}
\end{equation}
to yield the same stabilizer group and,  effectively,  the same code and fix the same code-space.
The symplectic Gram matrix is invariant under a symplectic transformation 
\begin{equation}
M\mapsto MS^T=\begin{pmatrix}
(S\bs{\xi}_1)^T\\ \vdots \\ (S\bs{\xi}_{2n})^T
\end{pmatrix}.
\end{equation}%
Note that,  unlike the change of basis,  a symplectic transformation generally changes the lattice but also leaves the symplectic Gram matrix invariant.  As a result,  symplectic transformations can serve as useful tool to adapt the code to the properties of the noise while keeping the same encoding rate,  as we will see in the following.  
When a GKP code is specified by its generator $\mathcal{S}=\mathcal{S}(M,0)$,  the corresponding phase sector transforms under symplectic transformation as \begin{equation}
\phi_M \mapsto \phi_{MS^T}.
\end{equation}

Any basis $M$ of the full rank lattice $\mathcal{L}$ partitions $\mathbb{R}^{2n}=\mathbb{R}^{2n}M$ into lattice translates of the \textit{fundamental parallelepiped} (note the row-vector convention adopted here),
\begin{equation}
\CP\lr{M}:=\lrc{\bs{x}^TM \,\big\vert \; \bs{x}\in [0,1)^{2n} }
\end{equation}
such that every point $\bs{x}^T \in \mathbb{R}^{2n}$ is uniquely associated to a lattice point $\bs{x}^T\in \bs{a}^TM + \CP\lr{M}$.  Equivalently,  one can choose to partition $\mathbb{R}^{2n}$ using the \textit{centered fundamental parallelepiped}
\begin{equation}
\CC\lr{M}:=\lrc{\bs{x}^TM \big\vert\; \bs{x}\in \left[-\frac{1}{2},\frac{1}{2}\right)^{2n} }.
\end{equation}
Any set that exactly partitions the space under lattice translations is called \emph{fundamental} and has volume 
\begin{equation}
\det\lr{\CL}:=|\det\lr{M}| =|\sqrt{\det\lr{A}}|=|\sqrt{\det\lr{G}}|
\end{equation}
where $G:=MM^T$ is the \textit{euclidean} Gram matrix associated to the generator $M$.  The determinant $\det\lr{G}$ is a lattice invariant known as its \textit{discriminant}.
The determinant (discriminant) can be geometrically interpreted as measure of inverse density of lattice points. 

A \textit{sublattice}\jc{it is called ``sublattice" and not ``sub-lattice" just like we have ``subsets" and not ``sub-sets"} $ \mathcal{L}' \subseteq \mathcal{L}$ is a subset of $\mathcal{L}(M)$ that is itself a lattice.  A basis for any $d-$dimensional sublattice of $\mathcal{L}(M)$,  where $d<2n$,  can be given by  
\begin{equation}
S=BM,
\end{equation}
where $B\in \mathbb{Z}^{d\times 2n}$ is an integer matrix of full row rank.
The \textit{symplectic dual} (simply ``dual'' in the following) of a lattice $\mathcal{L}$ is the lattice $\CL^{\perp}$ that consists of all vectors that have integer symplectic inner product with any vector from $\CL$%
\begin{equation}
\CL^{\perp} = \lrc{ \mathcal{\bs{\xi}}^\perp \in \mathbb{R}^{2n}\ |\ \lr{ \bs{\xi}^\perp }^T J \bs{\xi} \in \mathbb{Z}\ \forall \bs{\xi}\in \mathcal{L}   }.
\end{equation}%
Within the GKP code construction vectors $\mathcal{\bs{\xi}}^\perp \in \CL^\perp$  correspond to the displacement amplitudes that are associated to the centralizer of the stabilizers within the set of displacement operators.
A canonical choice of basis $M^{\perp}$ for the symplectic dual is specified by fixing $M^{\perp}$ to satisfy
\begin{equation}
M^{\perp}JM^T=I.
\end{equation} 
Since we will focus on full rank lattices,  for which $M$ is non-singular,  we obtain the canonical dual basis as \begin{equation}
M^{\perp}=(JM^T)^{-1}=M^{-T}J^T.\label{eq:canonical_perp}
\end{equation}
Together with the definition of $A$ it can be shown that 
\begin{equation}
M=AM^{\perp},
\end{equation}
that is,  the symplectic Gram-matrix $A$ describes how the sublattice $\CL=\CL\lr{M}\subseteq \CL^{\perp}=\CL\lr{M^{\perp}}$ associated to stabilizer operators embeds into the lattice associated to its centralizer.  
The \textit{dual quotient} (or \textit{glue group}) of $\CL$, $\CL^{\perp}/\CL$ thus lists the logically distinct displacements admitted by the GKP code $\mathcal{S}$.  Since we are interested in effective qubit encodings with logical dimension $d=2^k$,  these displacements form the effective logical Pauli group. 
\je{Maybe it is worth stating that not only the dimensions match, but that one actually
has a representation of the Pauli group.}\jc{thats what the text says...}

Finally,  the number of logically distinct centralizer elements associated to $\mathcal{S}$ is 
\begin{equation}
d^2=|\CL^{\perp}/\CL|=|\det\lr{M}|/|\det\lr{M^{\perp}}|=|\det{A}|=|\det\lr{M}|^2.
\end{equation}
One can verify this formula geometrically by imagining the partition of $\CP\lr{M}$ with patches $\CP\lr{M^{\perp}}$.
Under basis transformation of the direct lattice $M\mapsto UM$ the canonical dual basis transforms as
\begin{equation}
M^{\perp}\mapsto U^{-T}M^{\perp},
\end{equation}
and under symplectic transformations $M\mapsto MS^T $ we obtain
 \begin{equation}
M^{\perp}\mapsto M^{\perp} S^T=\begin{pmatrix}
(S\bs{\xi}^{\perp})_1^T\\ \vdots \\ (S\bs{\xi}^{\perp}_{2n})^T
\end{pmatrix}.
\end{equation}
The symplectic Gram matrix of the dual lattice can be shown to satisfy \begin{equation}
A^{\perp}:=M^{\perp}J\lr{M^{\perp}}^T=A^{-1}=\frac{1}{|\det\lr{A}|}\adj\lr{A},
\end{equation}
where $\adj\lr{A}$ is the \textit{adjugate} of $A$.
It is similarly common  to define the \textit{euclidean} dual of a lattice $\CL$ which we shall denote by $\CL^*$.  This is the lattice in $\mathbb{R}^{2n}$ consisting of all vectors with integer \textit{euclidean} inner product with every vector in $\CL$. 

The euclidean dual is more common than the symplectic one in the lattice theory literature,  and hence often simply called dual.  Similar to the symplectic case,  a canonical basis $M^*$ for the dual lattice can be fixed to satisfy %
\begin{equation}
M^* M^T=I 
\end{equation} %
or equivalently (since we only deal with full rank lattices) 
\begin{equation}
M^*=M^{-T}.
\end{equation}%
Since $J\in O(2n)$ is an orthogonal matrix,  we can observe that the symplectic dual is equivalent to the euclidean dual up to an orthogonal rotation.  In particular,  the distribution of lengths of vectors in $\CL^{\perp}$ is equal to  the distribution of lengths of vectors in $\CL^*$.  We will take advantage of this fact when analyzing the code distance of GKP codes in Section  \ref{sec:distance}.

We close this section by introducing the two most studied classes of GKP codes,  which we will use for illustrations throughout this manuscript.  The first class,  which we will refer to as \textit{scaled codes},  have been thoroughly examined in ref.~\cite{Harrington_Thesis} and build on symplectic \textit{self-dual} lattices.  A symplectic self-dual lattice is a symplectically integral lattice for which $|\det\lr{M}|=|\det\lr{A}|=1$ and,  consequently,  $\mathcal{L}=\mathcal{L}^{\perp}$.  The associated code-space is one-dimensional.

\jc{
Proof: if a lattice is symplectically integral we have $\CL \subseteq \CL^{\perp}$.  Asserting $|\det\lr{A}|=1$ implies that $|\CL/\CL^{\perp}|=1$,  that is,  $\CL=\CL^{\perp}$.
}

\paragraph{Scaled GKP codes.} 
A \textit{scaled GKP code} $\CL(M)$ is obtained by rescaling a symplectically self-dual lattice $\CL_0=\CL\lr{M_0} \subset \mathbb{R}^{2n}$ as
\begin{equation}
M=\sqrt{\lambda}M_0,\; \lambda \in 2\mathbb{Z}.
\end{equation}
The associated symplectic Gram matrix becomes $A=\lambda A_0$ and the dimension of the code-space is $d=|\sqrt{\det\lr{ \lambda A_0}}|=\lambda^n$.  By e.g.  choosing $\lambda=2$ this yields a code with $k=n\log_2\lr{\lambda}=n$.  
Let 
\begin{equation}
\lambda_1\lr{\CL_0}:=\min\lrc{\|\bs{x}\|, \, \bs{x}\in \CL_0}
\end{equation}
be the length of the shortest vector in $\CL_0=\CL_0^{\perp}$.  Note that throughout this manuscript $\|\cdot\|=\|\cdot\|_2$ denotes the euclidean \textit{2-norm}.
The rescaling implies  $\lambda_1\lr{\CL}= \lambda^{\frac{1}{2}}\lambda_1\lr{\CL_0}$ and $\lambda_1\lr{\CL^{\perp}}=\lambda^{-\frac{1}{2}}\lambda_1\lr{\CL_0}$.  For $\lambda>1$ the symplectic dual vector corresponding to $\lambda_1\lr{\CL^{\perp}}$ cannot be in $\CL$ and hence constitutes the shortest logically non-trivial displacement,  
the length of which 
\je{What length do you mean?}\jc{``the symplectic dual vector corresponding to $\lambda_1\lr{\CL^{\perp}}$"}
decreases with the number of logical dimensions that are squeezed into the code.
Prominent examples of scaled GKP codes that have already been discussed in ref.~\cite{GKP} are generated by
\begin{equation}
M_{\square}=\sqrt{2}I_2,
\end{equation}
also known as the \textit{square GKP code},  that encodes $k=1$ qubit into $n=1$ oscillator and
\begin{equation}
M_{\hexagon}=3^{-\frac{1}{4}}\begin{pmatrix}
2 & 0 \\ 1 & \sqrt{3}
\end{pmatrix},
\end{equation}
which is known as the \textit{hexagonal GKP code},  similarly with $k=n=1$ but with a different lattice geometry that requires a slightly larger displacement amplitude to implement a logical operator.  Here we have  $\lambda_1\lr{\CL^{\perp}}=3^{-\frac{1}{4}}$ compared to $\lambda_1\lr{\CL^{\perp}}=2^{-\frac{1}{2}}$ for the square GKP code.  Further symplectic self-dual lattices with larger $\lambda_1\lr{\CL_0}$ and a numerical procedure to find symplectic self-dual lattices are detailed in ref.~\cite{Harrington_Thesis}.

\paragraph{Concatenated square GKP codes.} 
The second class that will be extensively discussed throughout this work is that of concatenated codes built from the square GKP code and a qubit quantum-error correcting code.  Let $Q\subset \mathbb{Z}_2^{2n}$ be a set of binary symplectic vectors\footnote{That is $q^T J q' = 0\mod 2$ for all $q,\ q' \in Q$ (considering addition over reals or in $\mathbb{Z}_2$ does not make any difference).} such that $\mathcal{S}_Q=\lrc{Z^{q_1}_1Z^{q_2}_2\hdots X_{n+1}^{q_{n+1}}\hdots X_{2n}^{q_{2n}} \big\vert \forall \bs{q} \in Q}$ describes the stabilizer group 
of a $[\![n,k,d]\!]$ qubit quantum error correcting code.  
\je{This condition of having symplectic vectors is equivalent with the set being a stabilizer group, right?}\jc{yes.}
We can embed binary vectors in $\mathbb{R}^{2n}$ in the trivial way.  The lattice associated to the concatenated GKP code will be given by 
\begin{equation}
\CL = \Lambda\lr{Q}:=\lrc{\bs{x}\in \mathbb{R}^{2n}\big\vert \sqrt{2}\bs{x}\mod 2 \in Q }
\end{equation}
and equally describes an encoding of $k$ logical qubits into $n$ modes.
Such lattices are known in the literature as \textit{Construction A} lattices 
\cite{ConwaySloane}.\fa{does $\mathcal{Q}$ here actually describe a binary linear code? Or should we consider the X and Z parts separately? } \jc{I don't talk about the classical code here.  What $Q$ describes is the classical stabilizer group,  which you can also treat as a linear code.  It explicitly does not need to be CSS and we definitely should not consider the $X$ and $Z$ parts separately.}
%
Given $r=n-k$ symplectic vectors associated to a set of generators for $\mathcal{S}_Q$ that we stack row-wise into a generator $B_Q\in\mathbb{Z}_2^{r\times 2n}$,  we can write down the  generator for the concatenated code
\begin{equation}
M_\mathrm{conc}=\frac{1}{\sqrt{2}}\begin{pmatrix}
 B_Q  \\ \hline  
2I_{2n}
\end{pmatrix}.
\end{equation}
This basis is overcomplete and can be row-reduced to only consist of $2n$ generators.  
We discuss basis choices in the next section.

\section{Bases of GKP codes \label{sec:bases}}

The  $2n$ rows of the generator matrix $M$ of a GKP code stand in one-to-one correspondence with a minimal number of generators for the stabilizer group which are measured to deduce the error syndrome for recovery.  Hence the choice of basis $M$ strongly influences the physical implementation  of the code and certain basis choices may be beneficial over others.  Natural desiderata include the following properties. 

\paragraph{CSS property.} 
A \emph{CSS code} is known in qubit quantum error correction as a code for which the set of stabilizer operators disjointly separates into one that only contains Pauli-$Z$ type operators and one that only contains Pauli-$X$ type operators (or any other pairing).  In practice,  this has the advantage that ancillary qubits used to measure the stabilizers only need to come in two different initial states and only need to couple with their associated data qubits in one particular way,  i.e.,  via a $CNOT$ or $CZ$ operations.  Furthermore,  for sufficiently simple error models,  decoding algorithms only need to be tailored to the respective subcodes.  The most studied quantum error correcting codes are indeed CSS codes,  such as homological codes \cite{breuckmann_thesis} including the celebrated \emph{toric code} \cite{bravyi1998quantum} or planar versions of it \cite{tomita_2014},  which at the time of writing yield the best promise of protection from noise and allow for geometrically local implementations.
\jc{I dont get what ``expections are scarse" means.}
A natural generalization to GKP codes is to call \emph{CSS} a code whose lattice admits a basis that can be written as \cite{GKP}
\begin{equation}M=M_q \oplus M_p,\label{eq:CSSgen}
\end{equation}
 where $M_q,M_p \in \mathbb{R}^{n\times n}$ describe stabilizers that contain only $\bs{\hat{p}}$ and $\bs{\hat{q}}$ quadrature operators,  respectively.  Equivalentely,  the corresponding lattice can be decomposed into the direct sum of two disjoint sublattices
\begin{equation}
\CL = \CL_q \oplus \CL_p.
\end{equation}

A trivial example of GKP-CSS code is the single-mode square GKP code.  Concatenating a CSS qubit code with the square GKP code trivially preserves the CSS property. 
The decomposition eq.~\eqref{eq:CSSgen} holds concretely for the Hermite normal form  \cite{MicWar01, Micciancio01hnf} of a lattice basis when $M$ is (or can be rescaled to) integer, which is defined as follows.

 \begin{mydef}[Hermite normal form (HNF)]
Let $M$ be the basis for a lattice as described in the main text.  $M$ is in Hermite normal form if 
\begin{enumerate}
\item $M$ is upper triangular and of positive diagonal,  and
\item $0\leq M_{i,j}\leq M_{j,j}\; \forall i,  j $ ,  i.e.,  the pivot of each row is the largest element of its column.
\end{enumerate}
\end{mydef}

The Hermite normal form is special in that,  for integer $M$,  the HNF is unique \cite{HNF}
\je{This is meant only up to permutations of main diagonal elements?}\jc{no...}
and can be computed efficiently (see also ref.~\cite{CryptoLectureNotes}).  It can hence be used to verify the CSS property or compare whether two stabilizer groups are equivalent whenever we are given bases that can be scaled to be integer (and scaled back after applying the preferred HNF algorithm).  In particular, this applies to concatenated square GKP codes such that we can always check whether a given code is CSS by computing its HNF.

\paragraph{Short basis vectors.} Besides being CSS,  the homological qubit codes mentioned above are \textit{low-density-parity-check (LDPC)} codes,  meaning that they admit a choice of generators such that the support of each generator is bounded and each qubit only interacts with a  bounded number of generators. \jc{bounded=bounded by a constant should be clear for everyone who knows error correction...}
Here ``low'' typically means constant in the number of physical qubits.  The LDPC property can be generalized in two different ways.  Analogous to the qubit picture we can ask for the stabilizers of the GKP code to only act non-trivially on a low number of modes and that each mode be only touched upon by a low number of stabilizers.  It is straightforward to see that this property is inherited whenever a qubit LDPC code is concatenated with single mode bosonic codes such as the square- or hexagonal-GKP code and can more generally be quantified by counting the number of non-zero elements of a choice of generator $M$ per row/column.
On the other hand,  in the bosonic case,  it would also be natural to demand for the lattice basis vectors to be \textit{short} in the sense of the euclidean norm.  If one imagines that the stabilizer checks are implemented by coupling in and subsequently measuring  ancillary systems,  the length of the basis vectors depends both on the number of modes that each ancilla is coupled to as well as on the interaction-strength/-time required.  We will see shortly how a minimal basis (that is,  consisting of only $2n$ generators) with short vectors can be constructed for concatenated GKP codes.  In general,  when no extra structure is specified,  short bases can also be computed by means of \textit{lattice reduction} algorithms,  such as the \textit{Lenstra-Lenstra-Lov\'asz} (LLL) algorithm \cite{LLL},  which efficiently finds a short and almost orthogonal basis.

\paragraph{Orthogonality.}
On top of shortening the basis vectors,  the LLL algorithm also attempts to find a basis that is \textit{nearly} orthogonal.  A basis $M$ is orthogonal when $MM^T$ is diagonal and we say a lattice $\CL$ is orthogonal if such a basis exists.  

We will see later in Sections \ref{sec:distance},  \ref{sec:MLD} how orthogonality is useful when attempting to quantify the distance of a GKP code or for the task of decoding.  The single-mode square and rectangular GKP codes correspond to orthogonal lattices.  However,  when a qubit code is concatenated to the square-GKP code,  leading to 
Construction A lattices,  this can only be the case if the qubit stabilizer generators act disjointly on at most two quadratures at once \cite{deciding_orthogonality},  i.e.,  if the qubit stabilizer code decomposes as
\begin{equation}
Q=\bigoplus_i Q_i,  \hspace{1cm} Q_i=\{0,1\} \; \text{or} \; Q_i=\{00,11\}
\end{equation}
up to permutation of coordinates.
Since stabilizer codes require stabilizer generators to overlap non-trivially to achieve long-range entangled states,  we arrive at the following claim.

\begin{claim}[Orthogonality of concatenated GKP codes]\label{claim}
Concatenated GKP codes with long-range entangled code states have a non-orthogonal lattice $\CL$.
\end{claim}

Phrased differently,  this means that codes associated to orthogonal lattices cannot produce long-range entanglement. 
Finally,  we can also relate the length of stabilizer generators to the possibly admitted encoding rate.

\begin{them}[Hadamard's bound]\label{thm:hadamard}
Let $M$ be a GKP code's generator, $C=\max_i\; \|M_i\|$ with $M_i$ the $i$th row of $M$, and $d=2^k=\sqrt{\det\lr{A}}$ the logical dimension. We have 
\begin{equation}
k \leq 2n \log_2 C.
\end{equation}
\end{them}
\begin{proof}
We have $\det(M^{\perp})=\det(M)^{-1}$,  thus by $M=AM^{\perp}$ we have $|\det(A)|=|\det(M)|/|\det(M^{\perp})|=\det(M)^2$.  Hadamard's inequality implies 
\begin{equation}
|\det(M)|\leq \prod_{i=1}^{2n} \|M_i\| \leq C^{2n},
\end{equation}
which is also known as Hadamard's bound.
Taking the logarithm on both sides yields the result.
\end{proof}
This simple inequality based on Hadamard's bound constrains the encoding rate of a GKP code by the maximum length of the stabilizers.  For a fixed basis $M$,  $C$ quantifies the size of the displacement that needs to be measured to extract the syndrome,  which can loosely be interpreted as the ``experimental hardness'' to implement a code with stabilizer generators described by the rows of $M$.  $C$ is large whenever many modes are involved in a single stabilizer generator (compare with the LDPC property) or the displacement to be measured is large (requiring stronger coupling).  We note that the value of $C$ is basis-dependent,  while the LHS of the inequality is not.  Hence any fixed basis $M$ yields an upper bound on the encoding rate but the bound will only be saturated when all $M_i$ are orthogonal and of similar length.  We have already observed in Claim \ref{claim} that interesting concatenated codes do not yield orthogonal lattices,  such that in these cases the bound will not hold tightly.  

\subsection{Minimal bases for concatenated GKP codes}
As an example of an interesting application of lattice basis reduction,  we present a procedure to find a minimal basis for the lattice obtained as concatenation of the single square-GKP code with the 
\emph{Surface-17 code} in Fig.~\ref{fig:reduct_SurfGKP}.  This basis is derived by first noting that $2n$ stabilizers suffice to fully specify the code,  in contrast to the usual $3n-k$ generators of the union individual GKP stabilizers and concatenated stabilizers.  Starting from this set,  it is clear that removing any of the $r=n-k$ independent qubit stabilizers from the union of generators necessarily changes the dimension of the code space.  The redundant stabilizers,  that can be omitted without affecting the encoded subspace,  must hence be found among the single-mode GKP ones.  We identify the redundant single-mode GKP stabilizers as follows.  Squaring the concatenated qubit code stabilizers $S_i^{Z/X}$ yields a product of square GKP-code stabilizers,  such that not all GKP stabilizer overlapping on $\text{supp}\lr{S_i^{Z/X}}$ need to be included on that local region to generate the full local stabilizer group,  e.g.  in Fig.  \ref{fig:reduct_SurfGKP} one can obtain $S_2^p$ by multiplying $(S_3^X)^2$ with $S_3^p$,  such that locally $S_3^p$ and $S_3^X$ suffice to reconstruct the full $p-$type syndrome.  We note that the pattern of the included local GKP-code stabilizers depicted in Fig.~\ref{fig:reduct_SurfGKP} can be extended to larger surface code layouts.


In general,  given a concatenated square GKP code we obtain a procedure to find a minimal generating set as detailed in Figs.~\ref{fig:lattice completion},  \ref{fig:membership},  \ref{fig:membership_HNF}.  Let $B=B_Q \in \mathbb{Z}_2^{r\times 2n}$ be a minimal generator matrix for the qubit code.  We iterate through all GKP-stabilizer generators and check whether they are already member of the lattice generated by $B$.  If not,  we append this generator to the basis and proceed to the next.  This algorithm will terminate once $2n-r$ GKP stabilizer generators are added.  Importantly,  it preserves the qubit stabilizers and hence the LDPC property.  If the qubit code had a generating set with maximal weight $w$,  the  longest basis vector will be of length $\max\{\sqrt{{w}/{2}}, \sqrt{2}\}$.  Note that the membership test outlined in Fig.~\ref{fig:membership_HNF},  that applies when dealing with rational lattice (in which case the lattice can be rescaled to integer,  such that the HNF is unique),  can be executed in polynomial time.  
The above procedure can trivially be generalized to the concatenation of a qubit code with a different single-mode GKP code.  In particular if we use single mode GKP codes that are locally described by generator matrix $M_{\mathrm{GKP}}$,  the algorithm outlined in Fig.~\ref{fig:lattice completion} is adapted by iterating through the rows of $M_{\mathrm{GKP}}^{\oplus n}$ and matching the non-trivial entries in $B_Q$ to elements in $\lr{\CL\lr{M_{\mathrm{GKP}}}^{\perp}}^{\oplus n}$.  A \textit{python} implementation of the algorithm can be found under ref.~\cite{lattice_completion_code}.  

We conclude this section by noting a practical advantage following from the existence of minimal bases.  In the literature,  an error correction cycle for concatenated GKP codes typically contains two steps: first all single-mode GKP stabilizers are measured and the state is corrected to the corresponding sublattice.  Then all qubit level stabilizers are measured and the syndrome is fed to a modified qubit decoder,  possibly taking into account the analogue measurement information.  The existence of a minimal set of $2n$ lattice basis vectors implies that checking whether a given state is a code-state can be done with only $2n$ measurements.  This has to be contrasted with the usual $2n + n-k$ measurements in typical two-step schemes.  Using a minimal basis foregoes measuring the GKP stabilizers whose eigenvalues are already fixed by the results on the minimal set,  and can thus lead to a reduction of about $1/3$ in the number of (GKP) ancillas required for an error correction cycle. 
Although the eigenvalues of the redundant,  unmeasured,  GKP stabilizers can be deduced from the other measurements,  in the presence of measurement errors other strategies might be more effective.

\begin{figure}
\center
\includegraphics[width=.5\textwidth]{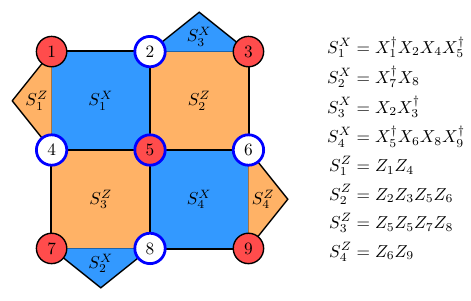}
\caption{The Surface-17 code with a minimal generating set.  Multi qubit stabilizer generators are indicated in the figure.  The minimal number of additional square-GKP stabilizer are given by $S^p_i=e^{-2i\sqrt{\pi}\hat{p}_i}$ on the modes labeled by red filled circles and  $S^q_i=e^{i2\sqrt{\pi}\hat{q}_i}$ on the vertices drawn in blue.  It suffices to measure only one type of square GKP generator on most modes,  except for mode 5,  where both $S_5^q$ and $S_5^p$ need to be measured.  The $\cdot^{\dagger}$'s are chosen such that all qubit stabilizers commute on the full bosonic Hilbert space \cite{toricGKP},  i.e.  such that $B_QJB_Q^T =0$. \protect\footnotemark We follow the convention for the stabilizers in ref.~\cite{surfGKP}.}  \label{fig:reduct_SurfGKP}
\end{figure}

\newtcolorbox{mybox}{colback=white, colframe=red!45!blue}

\begin{figure}
\footnotesize
\begin{minipage}[t]{.5\textwidth}{ $\;$}
\vspace{-1.54cm}
\begin{mybox}
\begin{algorithmic}[H]
\caption{Greedy lattice completion algorithm}\label{fig:lattice completion}
\Function{LatticeCompletion}{$B_Q$} 
\State set $B=B_Q$
\For{$i=1,\dots, 2n$ }
\If{ \textbf{not} Member$\lr{2\bs{e}_i^T; B}$}
\State $B\leftarrow \begin{pmatrix}
B  \\ \hline 
2\bs{e}_i^T 
\end{pmatrix}$ 
\EndIf
\EndFor
\State Return $\frac{1}{\sqrt{2}}B$
\EndFunction
\vspace{.13cm}
\end{algorithmic}
\end{mybox}
\end{minipage}\hspace{-.1cm}
\begin{minipage}[t]{.5\textwidth}
\begin{minipage}{\textwidth}
\vspace{-0.3cm}
\begin{mybox}
\begin{algorithmic}[H]
\caption{Membership test}\label{fig:membership}
\Function{Member}{$\bs{x};\, B$}
\State Return $\exists \bs{x} \in \mathbb{Z}^{2n}:\; \bs{x}^TB=2\bs{e}_i^T$
\EndFunction
\end{algorithmic}
\end{mybox}
\end{minipage}
\begin{minipage}{\textwidth}
\vspace{-0.3cm}
\begin{mybox}
\begin{algorithmic}[H]
\caption{Membership test (alternative)}\label{fig:membership_HNF}

\Function{Member}{$\bs{x};\, B$}
\State Return $\HNF\begin{pmatrix}
B  \\ \hline 
2\bs{e}_i^T 
\end{pmatrix}
=\HNF\lr{B}
$
\EndFunction

\end{algorithmic}
\end{mybox}
\end{minipage}
\end{minipage}
\end{figure}

\footnotetext[3]{For homological codes derived from \emph{Kitaev's quantum double} models \cite{Kitaev_2003},  these can be identified by considering local groups $\mathbb{R}$ and looking at the 
direction of the arrows.}

\section{Symplectic equivalence\label{sec:symplectic}}
By definiton of the symplectic Gram matrix,  it remains invariant under symplectic transformations of the GKP code $M\mapsto N=MS^T$.  Symplectic transformations are physically implemented by Gaussian unitaries which have a special role in bosonic quantum information processing,  for a number of reasons,  including the fact that they are often the easiest to implement in the laboratory.  As far as error correction is concerned,  the fact that they act linearly on the quadratures implies that potential errors,  such as small displacements on the modes,  remain bounded in their amplitude.  
By construction of the GKP code,  the logical Clifford group can be implemented using only symplectic transformations.\footnote{However,  not every symplectic transformation necessarily corresponds to a Clifford operation on the logical level.}  Here, we show that two GKP codes are equivalent under symplectic transformations if and only if their standard form (introduced below) is equal and show that every code is symplectically equivalent to a standard form where each stabilizer acts on a single mode.   This is a generalization of the symplectic equivalence of the square- and hexagonal GKP code noticed in ref.~\cite{Haenggli_2020}.
We say a symplectic Gram matrix is in standard form if 
 \begin{equation}
A=J_2 \otimes D,
\end{equation}
where $D$ is a positive diagonal $n \times n$ integer matrix with non-increasing diagonal elements.  A symplectically integral matrix $A$ can always be brought into standard form via basis transformation  \cite{GKP}. 

\begin{them}[Symplectic equivalence]\label{thm:symplectic_equiv}
Let $M,  N\, \in \mathbb{R}^{2n\times2n}$ be generators of full rank lattices $\CL\lr{M}, \,\CL\lr{N} \subset\mathbb{R}^{2n}$ .  Without loss of generality,  assume the bases $M,  N$ are chosen such that $A_M=MJM^T=J_2\otimes D_M$ and $A_N=NJN^T=J_2\otimes D_N$ are in standard form.  Then $M=NR^T$ for some symplectic matrix $R$ if and only if $D_N=D_M=:D$.
\end{them}

\proof{
From the definition of $A$ it is clear that $M=NS^T$ for some symplectic $S: S^TJS=J$ implies that $A_N=A_M$ or equivalently $D_N=D_M$.  
To prove the converse,  let 
\begin{equation}
Z:=I_2 \otimes D^{-\frac{1}{2}}, \hspace{.5cm} \tilde{M}:=ZM,\hspace{.5cm} \tilde{N}:=ZN.
\end{equation}
We have 
\begin{equation}
\tilde{M}J\tilde{M}^T=\tilde{N}J\tilde{N}^T=J,
\end{equation}
i.e.,  both $\tilde{M},  \tilde{N}$ are symplectic.  Since symplectic matrices are invertible,  we define $R^T:= \tilde{N}^{-1}\tilde{M}$ such that $\tilde{M}=\tilde{N}R^T$.  $R$ is symplectic since the symplectic 
 group is closed under transpose and inverse.  Left-multiplying by $Z^{-1}$ we obtain $M=NR^T$. \qed
}%

It is important to note that since the symplectic equivalence above is established using the bases $M,  N$ chosen such that $A_M,  A_N$ are in standard form,  one needs to be careful in accounting for possible phases in front of the associated stabilizer operators stemming from the basis transformation.  When the above assumptions of the theorem above are satisfied we obtain that the stabilizer group 
$ \mathcal{S}\lr{N,\,  0} $ can,  by symplectic transformation $R^T$, be transformed into
$ \mathcal{S}\lr{M,\,  \phi_{NR^T}}. $
I.e.  code states $\ket{\psi}$ of the code  $\mathcal{S}\lr{N,\,  0}$ can be transformed into states $U_{R^{-T}}\ket{\psi}$ by means of a Gaussian unitary (see eq.  \eqref{eq:symplDisplTrasf}) which are stabilized by phased displacement operators 
\begin{equation}e^{i \phi_{NR^T}\lr{M_i^T}} D\lr{M_i^T}, \, i=1,\hdots, 2n.
\end{equation}
Note that phases above are absent in the symplectic  transformation between a single-mode square GKP code and hexagonal GKP code shown in ref.~\cite{Haenggli_2020},  which is because each of these codes have symplectically even lattices,  such that the overall phase factors are trivial.  The same discussion on phase-assignement is relevant for the application of the following immediate corollaries.

\begin{cor}[Normal form of generators]
Let $M\in \mathbb{R}^{2n\times2n}$ describe the basis of a full rank lattice with integral $A=MJM^T$.  W.l.o.g. assume $M$ is chosen such that $A=J_2\otimes D$ is in standard form.  The code specified by the generator
\begin{equation}
N_{\square}=\oplus_{j=1}^n \sqrt{D_{j j}} I_2
\end{equation}
is symplectically equivalent to the one specified by $M$.  
\end{cor}

Note that here, we understand the direct sum as happening mode-wise,  such that it corresponds to the usual matrix direct sum if the quadratures are ordered as $\lr{\hat{q}_1,\hat{p}_1,\hat{q}_2,\hat{p}_2,\ldots}$ with symplectic form $J_{2n}'=I_n\otimes J_2$ but not if they are ordered as $\lr{\hat{q}_1,\hat{q}_2,\ldots,\hat{p}_1,\hat{p}_2,\ldots}$ in the standard choice $J_{2n}=J_2\otimes I_n$.  This convention will be understood in the following whenever clear from the context.  Note that $N_{\square}$ is diagonal,  so its stabilizers act each on a single quadrature at a time. 
Since $N_{\square}$ decomposes into a direct sum over each mode,  we can always prepare a code state of $M$ by locally preparing a code state of $N_{\square}$ and applying the corresponding symplectic transformation.

\fa{I am no longer sure this is true: what happens if the code state is a superposition?}\jc{we'd still obtain a valid code state -- To fix \emph{which} code state is prepared it would be necessary to check if $M^{\perp}$ admits a transformation into decomposing basis using the same $S$ which I think is the case,  which you can see by using $A^{-1},\; D^{-1}$ in the proof above.  }  
The same symplectically equivalent local decomposition can equally be done using the hexagonal GKP code on each local mode and is based on the fact that the square and hexagonal GKP code,  as scaled GKP codes with scaling factor $\lambda$,  have both $A_{\square}=A_{\mhexagon}=\lambda J_2$.  Similarly,  any other scaled GKP codes can also be used.  This is an example of \textit{code switching}.  Even more generally,  this implies code switching between two GKP codes $\CL\lr{M}$ and $\CL\lr{N}$ on the same number of modes is possible whenever $D_N=D_M$ in their respective standard bases.
Similarly,  we can also deduce the following corollary.

\begin{cor}[Normal form in prime dimensions]\label{thm:symplectic_equiv_cor}
For $d=|\det(M)|$ prime,  the lattice $\CL \subset \mathbb{R}^{2n}$ is symplectically equivalent to a code specified by
\begin{equation}
N'=\sqrt{d}I_2 \oplus I_2^{\oplus (n-1)}, \label{eq:symp_single}
\end{equation} 
\end{cor}
\begin{proof}
For a basis of $\CL$ in standard form,  we have %
\begin{equation}
d=\prod_{j=1}^{n}D_{j,  j}.
\end{equation}%
Since $D$ is a positive integer diagonal matrix,  if $d$ is prime,  the eigenvalues of $D$ are uniquely specified by $D=\diag(d, 1,1, \hdots)$ up to permutations.  Hence by Theorem \ref{thm:symplectic_equiv} we have that $\CL$ is symplectically equivalent to the code specified by eq.~\eqref{eq:symp_single}.
\end{proof}
From the proof above one can also see that the number of symplectically inequivalent classes of codes with given logical dimension $d$ corresponds to the number of different factorizations of $d$. For example, for two modes encoding two qubits there are two inequivalent choices: $D = \mathrm{diag}\lr{4,1}$ or $D = \mathrm{diag}\lr{2,2}$, corresponding to two symplectically inequivalent classes of codes.

\section{Distance of a GKP code}\label{sec:distance}

We need a benchmark to assess how well the logical content of a GKP state is protected from noise.  
In qubit codes,  a simple standard assumption is that noise is stochastic i.i.d.  for each physical qubit,  such that the likelihood of an error decreases exponentially with its support or weight.  This assumption makes the definition of a code distance as weight of the shortest non-trivial logical operator meaningful.  This model no longer makes sense in the bosonic setting.  In the bosonic setting,  a reasonable assumption about the underlying error model is that weak coupling to the environment results in effectively small displacements.  \jc{``Such errors give rise to Gaussian quantum channels if the
displacements themselves are drawn according to a Gaussian distribution \cite{GaussianChannel}. " sounds like a logical loop.  Also,  the justification for the distance is that we generally have an error model that favors small displacements -- it doesnt matter if those are drawn from a Gaussian distribution.}
Concrete examples are loss,  thermal noise \cite{Noh_Capacity},  and to a limited  extent,  finite squeezing errors \cite{Glancy_2006,  Tzitrin_2020, Terhal_2020}. 
For simplicity,  we assume a stochastic displacement noise model such that small displacements are more likely than large displacements,  this is in particular the case for the Gaussian displacement noise model
\begin{equation}
\mathcal{N}\lr{\rho}= \int d\bs{x}\; P_{\tilde{\sigma}}\lr{\bs{x}} D\lr{\bs{x}}\rho D^{\dagger}\lr{\bs{x}},\; \label{eq:noise_channel}
\end{equation}
 which we will analyse more in-depth later,  where 
\begin{equation}
P_{\tilde{\sigma}}\lr{\bs{x}}=G\lrq{0, \tilde{\sigma}^2 I_{2n}}\lr{\bs{x}}\propto e^{-\frac{\|\bs{x}\|^2}{2\tilde{\sigma}^2}},
\end{equation}
that displaces the quadratures of the state by some random amplitude drawn from a centered Gaussian distribution with variance $\tilde{\sigma}^2$.  Although this error model is widely used in the analysis of GKP codes \cite{toricGKP, surfGKP, Haenggli_2020} for its simplicity,  one needs to be careful to note that in real implementations of GKP error correction,  this is not a physically accurate model \cite{Terhal_2020, Conrad_2021} in general,  but only reproduces the correct measurement statistics in specific cases in practice,  such as when the finite squeezing error is applied to a perfect GKP state and interpreted as channel or when GKP states undergo a $0-$photon loss event.  The distance measure that we introduce in the following merely assumes that larger phase-space displacements are more likely than smaller ones,  which is a natural choice and is sufficiently meaningful to indicate the robustness of the code with respect to realistic noise sources as also demonstrated in recent experiments \cite{Campagne_Ibarcq_2020}.  
\je{I find the terminology of the ``shortest non-trivial logical operator'' somewhat unfortunate,
as strictly speaking it is a vector that is the shortest.}\jc{agree,  but I couldn't come up with anything better that doesn't sound way to bulky.  I think it is clear what is meant.}

\begin{mydef}[Euclidean distance of a GKP code]
The (euclidean) distance $\Delta$ of a GKP code given by lattice $\CL$  is the euclidean length of the shortest non-trivial logical operator,  i.e.,
\begin{equation}
\Delta=\Delta \lr{\CL} := \min_{0\neq \bs{x}\in \CL^{\perp}/\CL} \|\bs{x}\|. \label{eq:def_distance}
\end{equation}
\end{mydef}

For CSS GKP codes,  i.e.,  for which we have
$\CL=\CL_q \oplus \CL_p$ as specified earlier,  we can further distinguish 
\begin{equation}
\Delta_q := \min_{0\neq \bs{x}\in \CL_p^*/\CL_q} \|\bs{x}\|,
\end{equation}
\begin{equation}
\Delta_p := \min_{0\neq \bs{x}\in \CL_q^*/\CL_p} \|\bs{x}\|.
\end{equation}

\subsection{Tradeoffs and bounds for GKP codes}
Before discussing how the distance of GKP codes can be computed in general,  we present a few results that relate the distance of GKP codes to other code properties such as number of modes $n$ of the code,  the encoded logical dimension $d=2^k$ and size of stabilizers.

%
\begin{lem}[Distance bound]\label{thm:distance_bound}
Remember that $\lambda_1\lr{\CL^{\perp}}$ denotes the shortest non-zero vector in the dual lattice.  We have 
\begin{equation}
\Delta \geq \lambda_1\lr{\CL^{\perp}}.
\end{equation}
\end{lem}
\begin{proof}
Because the lattice vector  $0\neq \bs{x}\in \CL^{\perp}/\CL$ for which $\|\bs{x}\|$ is minimal is also in $\CL^{\perp}$ this holds trivially.
\end{proof}%
We have already seen in the previous section that for GKP codes obtained from scaling a symplectically self-dual lattice $\CL_0=\CL\lr{M_0}$ to $\CL=\sqrt{\lambda}\CL_0$ we have 
\begin{equation}
\Delta_\mathrm{scaled}=\lambda_1\lr{\CL^{\perp}}=\lambda^{-\frac{1}{2}}\lambda_1\lr{\CL_0},
\end{equation}
so the distance decreases while the size of the stabilizers and number of encoded logical dimensions increases.
In particular,  scaled codes are not straightforwardly  \textit{scalable} in the sense that we cannot increase the distance by adding generators while keeping the encoded logical dimensions fixed.   It has been shown \cite{Sarnak1994}, however   (see also ref.~\cite{Harrington_Thesis}), 
that symplectically self dual lattices in $\mathbb{R}^{2n}$ with $\lambda_1\geq 2\sqrt{\frac{n}{\pi e}}$ exist,  implicating that it is possible to find sequences of codes that can be scaled to \emph{achieve} rates $R=\log_2{\lambda}$,
\begin{equation} \lambda =\floor{\frac{\hbar}{2 \tilde{\sigma}^2}}
\end{equation}
for transmission through stochastic Gaussian displacement channels with variance $\tilde{\sigma}^2$~\cite{HarringtonPreskill}.
\fa{what's $\tilde{\sigma}$?}\jc{the noise variance for Gaussian channel,  another reason why we should introduce gaussian channel at the beginning of this section.  Also I would not say ``families".  Bavard and Sarnaks result only says that in each dimension you can find symplectic self dual lattices with sufficiently large distance in each dimension.  Harrington talks about a family of lattices but what he means is his parametrization used for brute-force search.  This is not what is viewed as a ``code family'' in QEC. }\fa{How about ``sequences''? I totally agree on introducing the Gaussian noise model, the point of my comment above was just to avoid subordinating the notion of distance to this specific noise channel, as it still makes sense for any channel where small displacements are more common than large ones.}\fa{I find it a little confusing that the rate how we intend it is a property of the code, independent of the noise model.  I guess the keyword here is ``achievable'', in the sense of channel coding. Idk maybe it's just me but it took me a while to figure out why the encoding rate here depends on the noise.}\jc{yes,  Harringtons proof is about achievable rates in the channel coding sense.}
On the other hand,  for the concatenation of a single mode GKP code into a $[\![n,k,d]\!]$ qubit quantum error correcting code,  we have
\begin{equation}
\Delta_{\text{conc}}\geq\sqrt{d}\Delta_{\text{loc}}, \label{eq:conc_dist}
\end{equation}
where $\Delta_{\text{loc}}$ is the distance of the local single mode GKP codes,  such as $\Delta_{\square}=2^{-\frac{1}{2}}$ for the square GKP code or $\Delta_{\mhexagon}=3^{-\frac{1}{4}}$ for the hexagonal GKP code.  Equality in eq.~\eqref{eq:conc_dist} holds when the code we consider is CSS or the local code is the hexagonal GKP code (such that all shortest non-trivial logical vectors have the same distance).  Eq.  \eqref{eq:conc_dist} is verified by decomposing a shortest representative non-trivial logical vector $\bs{L}=\oplus_{i=1}^n \bs{l}_i $ into $n$ local sub-blocks,  where $\bs{l}_i$ is a logical operator of the respective local code,  such that we have $\|\bs{L}\|^2\geq d \Delta_{\text{loc}}^2$.

 The distance of a concatenated code is typically strictly larger than $\lambda_1\lr{\CL^{\perp}}$,  because the shortest dual vectors would correspond to stabilizer displacements,  in particular if the input qubit stabilizer code is LDPC \fa{why for LDPC codes?}\jc{because the stabilizers have low constant weight while logicals are typically designed to span something that grows with the block size.}.  This is e.g.  the case for the concatenation of the square GKP code with $\textit{surface-17}$ (see Fig.  \ref{fig:reduct_SurfGKP}),  where we have $\Delta=\sqrt{{3}/{2}}$,  while  $\lambda_1\lr{\CL^{\perp}}=\lambda_1\lr{\CL}=1$.  

The length of the shortest vector in a lattice $\lambda_1\lr{\CL}$ is the first of $2n$ successive minima of the lattice.  Generally,  the $i$th successive minimum $\lambda_i\lr{\CL},  \, i\leq 2n$ is defined to be the smallest $r>0$,  such that $\CL$ contains $i$ linearly independent vectors of length at most $r$.  It holds that $\mu\lr{\CL}\geq {\lambda_{2n}\lr{\CL}}/{2}$,  where $\mu\lr{\CL}$ is the \textit{covering radius} of the lattice $\CL$,  i.e.,  the minimum radius $\mu$,  such that the union of closed balls $\mathcal{B}_{\mu}\lr{\bs{x}},  \,  \bs{x} \in \CL$ centred around each lattice point of $\CL$ cover the entire space $\mathbb{R}^{2n}$.  A related quantity is the \textit{packing radius} of the lattice $\rho\lr{\CL}={\lambda_1\lr{\CL}}/{2}$.  \\
Successive minima of the direct and dual lattice $\CL^*$ are related by so-called transference theorems,  in particular we have
\begin{align}
1\leq \lambda_1\lr{\CL}\lambda_{2n}\lr{\CL^*}\leq 2n. \label{eq:transference}
\end{align}
We have noticed earlier that $\CL^{\perp}$ and $\CL^*$ differ only by an orthogonal transformation,  therefore it holds that 
\begin{equation}
\lambda_i\lr{\CL^{\perp}}=\lambda_i\lr{\CL^*}
\end{equation}
and we can apply eq.~\eqref{eq:transference} to relate the distance to the length of stabilizer vectors.  

\begin{them}[Distance bound]
For a GKP code with lattice $\CL$,  distance $\Delta$ and maximal length $C$ of a basis vector for a fixed basis $M$ we have
\begin{equation}
\Delta\geq \lambda_1\lr{\CL^{\perp}}\geq \lambda_{2n}^{-1}\lr{\CL}\geq C^{-1},\label{eq:dist_C_bound_1}
\end{equation}
as well as
\begin{equation}
\Delta \leq \lambda_{2n}\lr{\CL^{\perp}}\leq \frac{2n}{\lambda_1\lr{\CL}}. \label{eq:dist_C_bound_2}
\end{equation}
\end{them}
\begin{proof}
The first bound follows immediately from Lemma \ref{thm:distance_bound}  and eq.~\eqref{eq:transference} by swapping the roles of $\CL^{\perp}$ and $\CL$,  which is possible because $\lr{\CL^{\perp}}^{\perp}=\CL$ and further from $C\geq \lambda_{2n}\lr{\CL}$.  Similar for the second bound.  
\end{proof}
These bounds indicate an intimate relation between the lengths of the stabilizers in $\CL$ and the distance of the GKP code,  which we will expand on further in the next section.  

\subsection{Symplectic transformations}

It is interesting to study how symplectic transformations change the code distance.  In particular,  since we have seen in Theorem \ref{thm:symplectic_equiv} that GKP codes with equal standard form are symplectically equivalent,  it is possible that for specific noise models a symplectic transformation of the stabilizers $M\mapsto MS^T$ can be used to improve the codes resilience to noise.  It is clear that orthogonal (symplectic) transformations $S$  satisfying $S^TS=I$ leave the distance invariant.

 This is not the case for squeezing.  In particular if we consider a uniform squeeze
\begin{equation}
S=\eta I_n \oplus  \eta^{-1} I_n, \quad \eta \in (0, \infty)
\end{equation}
applied to a CSS code,  we have
\begin{equation}
\Delta_q \mapsto \Delta_{q,\eta}= \eta^{-1}\Delta_q,\hspace{1cm} \Delta_p \mapsto  \Delta_{p,\eta} =\eta\Delta_p,
\end{equation}
so that we can adapt the squeezing to account for potential bias in the noise,  e.g.  when
\begin{equation}
P\lr{\bs{x}}=G[0,  \tilde{\sigma}^2 \lr{ \eta^{-1} I_n \oplus \eta I_n } ] \lr{\bs{x}}.
\end{equation}
For concatenated CSS GKP codes, we generally have 
\begin{equation}
\Delta_q=\sqrt{d_Z}\Delta_{q,  \mathrm{loc}},  \hspace{1cm} \Delta_p=\sqrt{d_X}\Delta_{p,  \mathrm{loc}},
\end{equation}
such that squeezing the local codes by $\eta$ is equivalent to increasing (decreasing) the upper level $X/Z-$ distances $d_{X/Z}$ by a factor of $\eta^2$.  On the other hand,  for natively unbiased noise it is  also possible to squeeze the local code and employ a qubit quantum error correcting codes tailored towards biased noise such that $\Delta_q = \Delta_p$ remains constant.  Although such a setup leaves the distance invariant,  it can still lead to improvements in the error correction procedure when dedicated decoders for the qubit error correcting codes are used,  as was recently demonstrated in ref.~\cite{Haenggli_2020}.

Corollary~\ref{thm:symplectic_equiv_cor} can also be used to derive an upper bound on the distance of a given code from symplectic equivalence as follows.  Let us first consider a code $\mathcal{L}$ with generator $M$,  encoding a single qubit within $n$ modes.  Suppose $M$ is in canonical form,  such that \begin{equation}
A = MJM^T = J_2 \otimes D\; \mathrm{with }\quad D = \mathrm{diag}\lrc{2,1,\hdots,1 }.
\end{equation} 
Corollary \ref{thm:symplectic_equiv_cor} implies that there exists a symplectic matrix $S$ such that $M = N_\square S^T$ with $N_\square $ the generator of the code in the corollary with logical dimension two.  Since $N_\square $ is diagonal,  the corresponding lattice $\mathcal{L}_\square$ is trivially orthogonal,  and so is the dual $\mathcal{L}_\square^\perp$. The shortest non-trivial logical operators are thus immediately found as $\bs{\eta}_{\square, 1}^T = \lr{1/\sqrt{2},\bs{0}_{2n-1}^T}$ and $\bs{\eta}_{\square,2}^T = \lr{0,  -1/\sqrt{2},  \bs{0}_{n-2}^T,}$.  We can now recall that commutation relations for the displacements are related to symplectic products of the corresponding phase space vectors and that $S$  preserves symplectic products.  Hence the transformation $S$ maps all points in $\CL^\perp_\square$ that correspond to stabilizers,  $\bs{x}\in\CL_{\square} \subset \CL^\perp_\square $, to direct lattice points $S\bs{x}\in\CL$,  which are stabilizers of $\CL$,  and all non-trivial logical operators to non-trivial logical operators.  Therefore, we can readily write down two logical operators,  $\bs{\eta}_1,  \bs{\eta}_2 \in \CL^{\perp}\setminus \CL$,  as $\bs{\eta}_j = S\bs{\eta}_{\square,  j} $ whose lengths are \begin{equation}
\norm{\bs{\eta}_j } = \norm{S\bs{\eta}_{\square,j}} \geq \Delta.
\end{equation}  
Note that the bound is not necessarily tight.  One can further relate the distance to the \emph{squeezing} contained in $S$.  Applying the Bloch-Messiah or Euler decomposition~\cite{braunsteinSqz} we can write $S = R_2KR_1$ with $R_j$ symplectic orthogonal matrices and $K = \mathrm{diag}\left\{ k_1,\ldots,k_n,1/k_1,\ldots,1/k_n \right\}$ with $k_l \in (0,\infty)$ denoting an effective squeezing operation.  Similar to
ref. ~\cite{oscillator_to_oscillator},  we denote by $\mathtt{sq}\lr{S}=\sqrt{\mathrm{EV}_{max}\lr{S^TS}}$ the root of largest eigenvalue of $S^TS$,  or equivalently,  the largest squeezing factor in the  Bloch-Messiah decomposition of the associated symplectic matrix.  Finally,  we obtain
\begin{equation}
\Delta\leq \|S\bs{\eta}_{\square,j}\| \leq \|K\|  \|\bs{\eta}_{\square,j}\| = \frac{1}{\sqrt{2}}\mathtt{sq}\lr{S}.
\end{equation} 
The bound presented above generalizes straightforwardly to codes with higher logical dimension,  such that we obtain the following bound. 

\begin{them}[Squeezing bound to the distance]\label{thm:distance_squeezing}
Let $\CL=\CL\lr{M}$ specify a GKP code with symplectic Gram matrix $A=J_2 \otimes D$ in its canonical form.  Further let $S$ denote the symplectic matrix that transforms between $M$ and the generator $N_{\square}$ as specified in 
Corollary \ref{thm:symplectic_equiv_cor}.  We have
\begin{equation}
\Delta \leq \sqrt{\max_j D_{j,j}}^{-1}\mathtt{sq}\lr{S}.
\end{equation}
\end{them}

This bound is interesting since it,  similar to the bound presented in ref.~\cite{oscillator_to_oscillator},  bounds a benchmark for error correcting capabilities of a code by the squeezing measure of the Gaussian unitary necessary to prepare the code state.  In ref.~\cite{oscillator_to_oscillator} that bound implies a no-go theorem on the existence of thresholds for oscillator into (many-) oscillator codes by the assumption that the squeezing measure $\mathtt{sq}\lr{S}$ is naturally bounded.  Here,  it is already known that code families with increasing distance exist,  e.g., given by concatenation of single mode GKP codes with the surface code which have distance scaling as $\Delta \propto n^{1/4}$,  where $n$ is the number of modes.  The crucial difference between our scenarios is,  however,  that ref.~\cite{oscillator_to_oscillator} focuses on the channel coding setup such that each code state necessarily is prepared by a Gaussian unitary encoding channel,  while in the general quantum error correction setting such as considered here,  code states can be formed by measurements of stabilizers (possibly  of a LDPC code) and correction via application of displacement operators,  forgoing the absolute constraint posed by the squeezing measure.  Similar to the argument in ref. \cite{oscillator_to_oscillator},  Theorem \ref{thm:distance_squeezing},  however,  does restrict the possible effective protection a code state prepared via a Gaussian unitary acting on code states obtained from the associated normal form as in Corollary  \ref{thm:symplectic_equiv_cor}.  Finally,  it is worth mentioning that,  by the Bloch-Messiah decomposition,  a non-trivial squeezing content  $\mathtt{sq}\lr{S}$ as specified in Theorem \ref{thm:distance_squeezing} is a necessary ingredient to obtain codes with non-orthogonal lattices and,  effectively,  non-trivial distance scaling.

\subsection{Computing the distance}
In general,  computing the distance of a GKP code given only the lattice $\CL$ or a corresponding basis $M$ is  computationally hard.  Given a non-trivial logical displacement $\bs{\xi}^{\perp}\in \CL^{\perp}$,  a candidate for the distance is computed by fully reducing $\bs{\xi}^{\perp}$ with respect to the lattice $\CL$,  that is one needs to find the length of the shortest \textit{representative} of $\bs{\xi}^{\perp}$ modulo the lattice $\CL$.  Equivalently,  it is necessary to solve the \textit{optimization-closest vector problem (optCVP)} defined as%
\begin{equation}
\optCVP\lr{\bs{x},  \CL}:= \min_{\bs{\xi}\in \CL} \|\bs{x}-\bs{\xi}\|=\dist\lr{\bs{x},\CL}. \label{eq:def_optCVP}
\end{equation} 
A related problem is the \textit{(search) closest vector problem (CVP)} that aims at identifying the closest lattice vector to the input
\begin{equation}
\CVP\lr{\bs{x},  \CL}:= \argmin_{\bs{\xi}\in \CL} \|\bs{x}- \bs{\xi}\| \label{eq:def_CVP},
\end{equation} %
which given inputs $\bs{\xi}^{\perp}$ and $\CL$,  outputs the shortest representative of $\bs{\xi}^{\perp}$.  These problems are known to be NP-hard.   
\je{Interestingly, this is apparently known to be NP-hard for 
randomized reductions only, and only the 
 problem with respect to the uniform norm is known to be NP-hard.}
Naively, one may be tempted to try to identify the shortest logical vectors by going through the rows of $M^{\perp}_i=\bs{\xi}^{\perp\, T}_i$ and rounding it element-wise to the closest element in $\CL$,  i.e.%
\begin{equation}
\lr{\bs{\xi}^\perp_i/M}^T=\bs{\xi}^{\perp\, T}_{i}-\nint{\bs{\xi}^{\perp \, T}_{i} M^{-1}}M, 
\end{equation}%
where $\nint{x}$ denotes the closest integer to $x$,  \footnote{acting element-wise} in other words identifying the shortest logical representatives with non-zero rows of %
\begin{equation}
M^{\perp}/M=M^{\perp}-\nint{ M^{-T}J^{-1}M^{-1} }M=M^{\perp}-\nint{ A^{-1} }M=\left(I-\nint{A^{-1}}A\right)M^{\perp}.\label{eq:Mperp_mod_M2}
\end{equation}
Indeed,  we have found numerically for the concatenated GKP code built on Surface-17 with $M$ given in its HNF that this rounding method outputs the shortest logical representatives.  This is, however, a lucky coincidence.  It is known that the rounding method for solving CVP \cite{babai} perform poorly in general and only returns an accurate solution if the lattice basis is sufficiently well-behaved,  e.g.  when the rows of $M$ are orthogonal.  This is consistent with the fact that in general the solution of CVP is also hard to \textit{approximate} better than up to an exponential factor in the dimension.  In addition,  it is not \textit{a priori} granted that the basis for the symplectic dual $M^{\perp}$ consists of those vectors  that correspond to the shortest logical displacements only up to stabilizers.  Hence while this method allows to upper bound the distance $\Delta$ of a general GKP code by%
\begin{equation}
\Delta\leq\tilde{\Delta}:=\min_i \|\lr{M^{\perp}/M}_i\|,
\end{equation}
the r.h.s.  is not guaranteed to yield an accurate estimate of the code distance.
\je{Yes, this is true.  But even before coming to the theta functions, convex relaxations should
perform better than the above and they are rigorous.}\jc{I have no result on this not a concrete idea how we can find a bound/approximation using convex relation,  what do you have in mind?}
Here, we propose a different strategy to compute the distance that is feasible in particular for scaled GKP codes and concatenated codes.  
We introduce the \textit{theta function} of a lattice \cite{ConwaySloane}%
\begin{equation}
\Theta_{\CL}(z)=\sum_{\bs{x}\in \CL} q^{\bs{x}^T\bs{x}}=\sum_{\delta \in \CD } N_{\delta} q^{\delta}, \label{eq:def_theta}
\end{equation}%
where $q=e^{i\pi z}$.  $\CD=\lrc{\|\bs{x}\|^2_2,  \,  \bs{x}\in \CL}$ is the set of squared distances of $\CL$.  We have also introduced the number of lattice vectors of a given length %
\begin{equation}
N_{\delta}=\#\lrc{\bs{x}\in \CL:\,  \bs{x}^T\bs{x}=\delta}.\label{eq:def_N}
\end{equation}%
For integral lattices, i.e.,  when the corresponding (euclidean) Gram matrix satisfies $G=MM^T \in \mathbb{Z}^{2n\times 2n}$, we can set $\CD=\mathbb{N}$ and  $N_m$ is given by the number of integer solutions $\bs{n}\in \mathbb{Z}^{2n}$ to the equation $\bs{n}^T G \bs{n}=m$,  which is an example of a \textit{Diophantine equation}.  We call the pair $(\CD,  N_{\delta})$ the \textit{distance distribution} of the lattice $\CL$.
The theta function converges and is holomorphic for $\Im(z)>0$.  The first summands are given by %
\begin{equation}
\Theta_{\CL}(z)=1+\tau q^{\lambda_1}+\hdots,
\end{equation}%
where $\lambda_1=\lambda_1\lr{\CL}$ is the length of the shortest vector of the lattice and $\tau$ is known as the \textit{kissing number},  the number of minimal-length vectors of $\CL$.
It is known that the theta function for the (euclidean) dual of a lattice $\CL\subset \mathbb{R}^{2n}$ is given by%
\begin{equation}
\Theta_{\CL^*}(z)=\det\lr{\CL}\lr{\frac{i}{z}}^{n}\Theta_{\CL}\lr{-\frac{1}{z}}, \label{eq:theta_functional}
\end{equation}%
which follows from the Poisson summation formula (see also ref.~\cite{theta_course}).  In particular,  since we have seen that $\CL^*$ and $\CL^{\perp}$ only differ by an orthogonal transformation, which does not change the length of lattice vectors, and hence the theta function,  we also have 
\begin{equation}
\Theta_{\CL^{\perp}}(z)=\Theta_{\CL^*}(z),
\end{equation}
such that we can immediately obtain the theta function of  $\CL^{\perp}$ whenever the the theta function of the direct lattice $\Theta_{\CL}(z)$ is known.  
By definition of the theta function,  the distance of a GKP code specified by $\CL$ is given by the smallest non-zero power of $q$ in%
\begin{equation}
Q_{\CL}(z):=\Theta_{\CL^{\perp}}(z)-\Theta_{\CL}(z)=N_{\Delta^2}q^{\Delta^2}+\hdots,  \label{eq:distance_theta}
\end{equation}%
and furthermore since we have that $\Theta_{\CL}(z)$ is uniquely determined by the distance distribution $(\CD,  N_{\delta})$ of $\CL$,  which by eq.~\eqref{eq:theta_functional} also fully specifies $\Theta_{\CL^{\perp}}(z)$.  Note that  $Q_{\CL}(z)>0$ since $\CL \subseteq \CL^{\perp}$.  
We arrive at the following insight.
\begin{them}[Code distance is specified by the distance distribution]\label{thm:distance_dist}
The distance of a GKP code specified by $\CL$ is uniquely determined by its distance distribution $(\CD,  N_{\delta})$.
\end{them}%
The theta function of a scaled lattice is %
\begin{equation}
\Theta_{\sqrt{\lambda}\CL_0}(z)=\Theta_{\CL_0}(\lambda z).  
\end{equation}%
Many expressions of theta functions of symplectically self dual lattices,  in particular those that also correspond to euclidean self-dual lattices are known in the literature \cite{ConwaySloane} such that their corresponding distances can e.g.  be estimated from a logarithmic fit for small $q\ll 1$ of eq.~\eqref{eq:distance_theta},  or by expressing $Q$ in a basis for which the distance distribution is known. 
We have verified that this method works numerically for some small examples including single mode scaled codes and the concatenated GKP-repetition code using the \textit{weight-distribution} as outlined below. 

\fa{I think it would be nice to have a couple of examples where you show you can correctly estimate the distance from the fit.}

For concatenated (square) GKP codes,  where $\CL=\Lambda\lr{Q}$ is given by a Construction A  lattice,  we can express the theta function using the \textit{weight enumerator}.  We first introduce the \textit{weight distribution} of a linear code $Q$,  which is given by the numbers $\lrc{A_i}_{i=0}^{2n}$ of codewords $q$ in $Q \subset \mathbb{Z}_2^{2n}$ with Hamming weight $\wt(q) = i$.  Crucially,  the weight distribution $\lrc{A_i}_{i=0}^{2n}$ here refers to the Hamming-weight distribution of the \textit{symplectic representation} of the qubit stabilizers.
The weight enumerator is given by
\begin{equation}
W_Q\lr{x, y}=\sum_{q\in Q} x^{2n-\wt(q)}y^{\wt(q)} = \sum_{i=0}^{2n}A_i x^{2n-i}y^i.  
\end{equation}
Using this definition,  we can express the theta function of a construction A lattice \cite{ConwaySloane} by straightforward computation%
\begin{equation}
\Theta_{\Lambda(Q)}=W_Q\lr{\theta_3(2z), \theta_2(2z)},
\end{equation}%
where%
\begin{align}
\theta_3(z)&=\Theta_{\mathbb{Z}}(z)=\sum_{m\in \mathbb{Z}}q^{m^2}, \\
\theta_2(z)&=\Theta_{\mathbb{Z}+\frac{1}{2}}(z)=\sum_{m\in \mathbb{Z}}q^{\lr{m+\frac{1}{2}}^2}. 
\end{align}%
Since we see that we can find the theta function of a Construction A lattice corresponding to a concatenated code by means of the (Hamming) weight distribution of its (qubit) stabilizer group,  we have similar to Theorem \ref{thm:distance_dist} and by eq.~\eqref{eq:conc_dist},

\begin{cor}[Distance of CSS qubit stabilizer code from weight distribution \je{Ok?}]\label{cor:qubit_distance_weights}
The distance $d$ of a CSS qubit stabilizer code is fully determined by the weight distribution $\lrc{A_i}_{i=0}^{2n}$ of its stabilizers.
\end{cor}

We note that this corollary also follows from a quantum version of the weight enumerators defined by Shor and Laflamme \cite{ShorLaflamme} and Rains \cite{Rains},  which through their immediate relationship to the Quantum error correction conditions \cite{Knill_2000} impose strong restrictions on possible quantum error correcting codes. 
Generally,  when $\CL$ is integral (rational),   $\Theta_{\CL}\lr{z}$ is a modular form\footnote{if the lattice can be scaled to be integral,  we consider the scaled version before doing the analysis and scale back accordingly afterwards.},  i.e.,  it remains invariant under certain M\"obius transformations 
\begin{equation}
z\mapsto \frac{az+b}{cz+d} 
\end{equation}
and as a result can be expressed in a finite basis depending on the specific symmetry group of the theta function \cite{modularforms_approach, theta_course, ConwaySloane}.  This avenue is interesting for the explicit computation of distances but extends beyond the scope of the present article.  We thus defer further analysis of a modular forms approach to finding distances to future work.%

\section{Maximum likelihood decoding GKP codes}\label{sec:MLD}

In this section,  we discuss active quantum error correction and \emph{maximum likelihood decoding (MLD)} for general GKP codes using the tools discussed in this work.  The introduction to MLD follows closely the discussion in ref.~\cite{toricGKP}.
We assume a stochastic Gaussian displacement noise channel as specified in eq.~\eqref{eq:noise_channel} with variance $ \tilde{\sigma}^2$. \jc{already specified that this is Gaussian} For comparison with the literature,  when the displacement operators are defined by a more ``standard" convention without the overall constant $\sqrt{2\pi}$,  this corresponds to a physical variance of $\sigma^2=2\pi  \tilde{\sigma}^2$. %
Upon sampling an error $\bs{e}$ and measuring the stabilizers,  a syndrome vector of the form%
\begin{equation}
\bs{s}(\bs{e})=MJ\bs{e} \mod 1
\end{equation}%
is obtained as the phases of the eigenvalues of the stabilizer generators $\lrc{D\lr{M_i^T}}_{i=1}^{2n}$ when acting on a code state vector $\ket{\psi}$ displaced by an error vector $\bs{e}$,
\begin{equation}
D\lr{M_i^T}D\lr{\bs{e}}\ket{\psi}=e^{i2\pi M_i J\bs{e}} D\lr{\bs{e}}\ket{\psi},
\end{equation}
where $M_i=\bs{\xi}_i^T$ is the i'th row of $M$.

Since we are dealing with full rank lattices,  given the syndrome,  we can assign a pure error %
\begin{equation}
\bs{\eta}(\bs{s})=(MJ)^{-1}\bs{s} \label{eq:eta}
\end{equation} %
that has the same syndrome $\bs{s}$ as $\bs{e}$ as ``initial guess" for the correction.  
The net displacement of amplitude $\bs{e}-\bs{\eta}$ is necessarily in the dual lattice, and hence maps the code space to itself, but it may be logically non-trivial.  It might thus be necessary to modify the initial guess by a logical vector (leaving the syndrome invariant), such that $\bs{e}-\bs{\eta}$ is a stabilizer.  Note that any $\bs{\eta} + \bs{l}$ with $\bs{l}\in \CL^\perp$ will lead to the same syndrome, so the choice of the additional logical displacement cannot be inferred from the syndrome alone.  Rather,  it is necessary to take the noise model into account. 

To find the appropriate logical post-correction, for every $\bs{\xi}^\perp \in \CL^\perp/\CL$ we evaluate the probabilities that,  given syndrome $\bs{s}$,  the actual error is stabilizer equivalent to $\bs{\eta}(\bs{s})+\bs{\xi}^{\perp}$, which is given by
\begin{equation}
P([\bs{\eta}(\bs{s})+\bs{\xi}^{\perp}] | \bs{s}) = P^{-1}(\bs{s}) \sum_{\bs{\xi} \in \mathcal{L}} P_{\tilde{\sigma}}(\bs{\eta}(\bs{s})+\bs{\xi}^{\perp}+\bs{\xi}), \label{MLD_1}
\end{equation} 
where $\lrq{\bs{x}} = \left\{ \bs{x} + \bs{\xi}, \bs{\xi}\in \CL\right\} $ and  $P_{\tilde{\sigma}}$ is as specified in eq.~\eqref{eq:noise_channel}. %
This can be rewritten as%
\begin{align}
P([\bs{\eta}(\bs{s})+\bs{\xi}^{\perp}] | \bs{s}) &= P^{-1}(\bs{s}) \sum_{\bs{\xi} \in \mathcal{L}+\bs{\eta}(\bs{s})+\bs{\xi}^{\perp}} P_{\tilde{\sigma}}(\bs{\xi}) \\
&=\sqrt{2\pi\tilde{\sigma}^{2n}}^{-1}  P^{-1}(\bs{s})  \Theta_{\mathcal{L}+\bs{\eta}(\bs{s})+\bs{\xi}^{\perp}}\left(\frac{i}{2\pi \tilde{\sigma}^2}\right), \label{MLD_2}
\end{align}%
proportional to the theta series of the \textit{packing} $\CP=\CL+\bs{\eta}(\bs{s})+\bs{\xi}^{\perp}$ evaluated in $z = \frac{i}{2\pi \tilde{\sigma}^2}$ \footnote{ $\CP$ as the translate of a lattice $\CL$ is formally not a lattice,  in particular $\CP$ may not contain the origin.} .
Let $\lr{\tilde{\CD},\, N_{\tilde{\delta}}}$ denote the distance distribution of $\CP$.
We can write the coset probabilities above in a \textit{small error} or \textit{``low temperature expansion"}. \footnote{Note that in contrast to the usual literature we do not distinguish a Nishimori line,  such that these are the same things.}%
\begin{equation}
 \Theta_{\mathcal{L}+\bs{\eta}(\bs{s})+\xi^{\perp}}\left(\frac{i}{2\pi \tilde{\sigma}^2}\right) = \sum_{\tilde{\delta} \in \tilde{\mathcal{D}}} N_{\tilde{\delta}} q^{\tilde{\delta}}. \label{MLD_3}
\end{equation}%
evaluated at $q = \exp\lr{ -1/2\tilde{\sigma}^2 }$.  Finally,  MLD decoding is implemented by applying the total correction %
\begin{equation}
\overline{\bs{\eta}}=\bs{\eta}(\bs{s}) + \argmax_{\bs{\xi}^{\perp} \in \mathcal{L}^{\perp}/\mathcal{L}} P([\bs{\eta}(\bs{s})+\bs{\xi}^{\perp}] | \bs{s}).
\end{equation} 
Methods to evaluate or sample the theta function of a lattice have been discussed in the literature \cite{modularforms_approach}.   Here, we find that,  given an algorithm that reliably estimates the theta function of a packing,  we can also approximately solve MLD decoding.


\paragraph{Minimum energy decoding (MED).}
In the limit $\tilde{\sigma},  q \rightarrow 0$ the sum \eqref{MLD_3} becomes sharply distributed around solutions with minimal $\tilde{\delta}$.  That is,  the bulk of the sum \eqref{MLD_3}  is determined by
\begin{equation}
\argmin_{\xi \in \mathcal{L}} \|\xi + \bs{\eta}(\bs{s})+\xi^{\perp}\|.
\end{equation}
The MLD correction then becomes
\begin{align}
\overline{\bs{\eta}}
&=\bs{\eta}(\bs{s}) +\argmin_{\xi^{\perp} \in \mathcal{L}^{\perp}/\mathcal{L}}  \;  \min_{\xi \in \mathcal{L}} \|\bs{\eta}(\bs{s})+\xi^{\perp}+\xi \| \\
&=\bs{\eta}(\bs{s}) + \argmin_{\xi^{\perp} \in \mathcal{L}^{\perp}}   \|\bs{\eta}(\bs{s})+\xi^{\perp}\| \\
&=\bs{\eta}(\bs{s}) - \argmin_{\xi^{\perp} \in \mathcal{L}^{\perp}}   \|\bs{\eta}(\bs{s})-\xi^{\perp}\|. \label{ME_1}
\end{align}%
The last line Eq~\eqref{ME_1} shows that this immediate version of MED requires solving CVP on the lattice $\CL^{\perp}$.  In the special case of a concatenated code,  $\CL=\Lambda\lr{Q}$,  we have that $\CL_\mathrm{GKP}:=\lrc{\bs{x}:\, \sqrt{2}\bs{x} \mod 2 =0}$ with generator matrix $M_\mathrm{GKP}=\sqrt{2}I_{2n}$ is a sublattice  of $\CL$ and its dual contains $\CL^{\perp}\subseteq \CL_\mathrm{GKP}^{\perp}$. 
In the literature,  a common approach to solve MLD for concatenated GKP codes is to apply the above outlined CVP decoder to the superlattice $\CL_\mathrm{GKP}^{\perp}$,  which leaves a net qubit level syndrome to be decoded with a standard qubit decoder \cite{Haenggli_2020, surfGKP, toricGKP}.  Solving CVP on $\CL_\mathrm{GKP}^{\perp}$ is easy since it is an orthogonal lattice and one can straightforwardly apply the rounding trick outlined in the previous section. 
However,  as we see below in such a two-step approach the initial guess for the correction on the sublattice $\CL_\mathrm{GKP}$ is arbitrary.

\paragraph{MLD decoding concatenated codes.} 
For concatenated (square) GKP codes,  we dissect the initial correction in two components %
\begin{equation}
\bs{\eta}(\bs{s})=\bs{\eta}^\mathrm{GKP}+\bs{\eta}^Q,
\end{equation}  such that %
\begin{equation}
\bs{\eta}^\mathrm{GKP}=(M_{GKP}J)^{-1}\bs{s}_{\mathrm{GKP}}
\end{equation}%
reproduces the syndrome detected by the single-mode GKP stabilizers $\CL_\mathrm{GKP}$ and $\bs{\eta}^Q \in \CL_\mathrm{GKP}^{\perp}=
\mathbb{Z}^{2n}/\sqrt{2}$ is such that the composite initial guess is consistent with the overall syndrome.  For notational clarity we suppress the explicit dependency on the syndrome $\bs{s}$.  %
The relevant theta function evaluates to%
\begin{equation}
 \Theta_{\mathcal{L}+\bs{\eta}(\bs{s})+\bs{\xi}^{\perp}}\lr{z}=\sum_{\bs{x} \in Q} \prod_{i=1}^{2n} \sum_{v_i \in \mathbb{Z}} q^{\lrq{\lr{\bs{x}+\sqrt{2}\bs{\eta}^\mathrm{GKP}+\sqrt{2}\bs{\eta}^{Q}+\sqrt{2}\bs{\xi}^{\perp}}_i+2v_i}^2},\label{MLD_4}
\end{equation}%
where we see that we can w.l.o.g.  reduce $\bs{\eta}^Q,  \bs{\xi}^{\perp} \in \frac{1}{\sqrt{2}}\mathbb{Z}_2^{2n} $ and consider the sum in the exponent $\bs{x}+\sqrt{2}\bs{\eta}^{Q}+\sqrt{2}\bs{\xi}^{\perp}$ to be carried out modulo $2$. \footnote{Any other choice just permutes the sum.} 
We interpret the individual terms
\begin{equation}
f\lr{b_i | \eta^\mathrm{GKP}_i}\equiv \sum_{v_i \in \mathbb{Z}} q^{\lr{\sqrt{2}\eta^\mathrm{GKP}_i+b_i+2v_i}^2}=e^{-\frac{1}{2\tilde{\sigma}^2}\lr{\sqrt{2}\eta^\mathrm{GKP}_i+b_i}^2}\theta_3\lr{\frac{i}{\tilde{\sigma}^2}\lr{\sqrt{2}\eta^\mathrm{GKP}_i+b_i} \Big\vert \frac{2i}{\pi \tilde{\sigma}^2} }
\end{equation}%
as proportional to an indicator for the likelihood of a local bit value $b_i=x_i+\sqrt{2}\eta_i^{Q}+\sqrt{2}\xi_i^{\perp} \in \mathbb{Z}_2^{2n}$
where we have used the definition%
\begin{equation}
\theta_3\left(\xi | z \right)=\sum_{m\in \mathbb{Z}} e^{2mi\xi+\pi i z m^2}, \,\ \Im(z)>0.
\end{equation}%
Overall,  eq.~\eqref{MLD_4} can then be rewritten as%
\begin{equation}
 \Theta_{\mathcal{L}+\bs{\eta}(\bs{s})+\xi^{\perp}}\left(\frac{i}{2\pi \tilde{\sigma}^2}\right)=\sum_{\bs{x}\in Q} \prod_{i=1}^{2n} f\lr{b_i | \eta^{GKP}_i},\label{MLD_TN}
\end{equation}%

reproducing a form for the MLD probability similar to those provided in ref.  \cite{Haenggli_2020} for the surface-GKP code.  It is clear that this derivation can be carried out similarly for the concatenation with other single-mode GKP codes or more general codes,  such as those derived from glued lattices discussed in the next section.  The MLD probability in eq.  \eqref{MLD_TN} can be expressed using a local tensor network whenever the underlying qubit code is LDPC as has been demonstrated in refs.  \cite{Bravyi_2014,  Haenggli_2020}.  Although qualitatively similar,  the effective local error probabilities in this formulation are not exactly the same as in refs.~\cite{surfGKP,  Haenggli_2020} since normalization constants $\sum_{b_i \in \mathbb{Z}_2} f\lr{b_i | \eta^\mathrm{GKP}_i} $ cannot be immediately pulled in from the global normalization $P(s)$ below each $f\lr{b_i | \eta^{GKP}_i}$ term.  
Besides the generality of this formulation of the MLD for GKP codes,  it would be interesting to test whether it also provides a quantitative advantage in decoding.  We examine this question more in-depth in a forthcoming publication. 



\section{Other constructions of GKP codes \label{sec:constructions}}
In this section we present two constructions of GKP codes that go beyond the scaled and concatenated GKP codes that are based on lattice glueing and the lattice tensor product.  
\je{This is the possibly most interesting direction, in the long term, we should maybe stress this earlier.}
We expect that such alternative constructions allow for more flexibility in designing codes in the future,  in particular when noise varies in strength for different modes or when no physical locality restrictions are present such as in photonic architectures.  The following constructions are adapted from the discussion of glued (euclidean) lattices and the lattice tensor product in refs.~\cite{ConwaySloane,  gannon_thesis}.

\subsection{Glued codes beyond concatenation}
We can understand concatenated GKP codes as GKP codes built on certain types of \textit{glued lattices},  which include those obtained through Construction A.
\je{Capitalized.}
 Here, we give a brief overview of glueing theory for symplectic lattices, leaning on the description in ref.~\cite{ConwaySloane}.  We begin by dissecting a (glued) symplectic lattice with symplectic sublattice in a top-down approach to understand its structure and then move to a bottom-up approach to construct a glued lattice from a base lattice by appending an appropriate glue group.  
Let us assume that we have a $2n-$dimensional symplectic lattice $\mathcal{L}$ that has a (symplectic) sublattice $\CL_0$ with direct sum structure 
\begin{equation}
\CL_0=\bigoplus_{i=1}^k \mathcal{L}_i.
\end{equation} 
Vectors  $\bs{v} \in \mathcal{L}$ can be written as 

\begin{equation}\bs{v} =\sum_i \bs{v} _i,\label{eq:vGlue}
\end{equation} 
where $\bs{v} _i \in  \mathbb{R} \otimes \mathcal{L}_i$.  The symplectic inner product of any $\bs{v} _i$ with any vector of $\CL_i$ is integer,  such that it can be concluded that $\bs{v} _i\in \mathcal{L}_i^{\perp}$.  Moreover we can add to any $\bs{v} _i$ a vector from $\mathcal{L}_i$ without changing the fact that $\bs{v}$  has integer symplectic inner product with any other vector of $\CL$.  It hence suffices to demand $\bs{v} _i \in \mathcal{L}_i^{\perp}/\mathcal{L}_i$.  Such vectors are called \textit{glue vectors} for $\mathcal{L}_i$,  which in the coding language correspond to logical representatives of a local code.  $\mathcal{L}_i^{\perp}/\mathcal{L}_i$ is also known as the (symplectic) dual quotient or \textit{glue group} for $\mathcal{L}_i$.  We can thus obtain symplectic lattices from a \textit{base lattice} $\CL_0=\bigoplus_{i=1}^k \mathcal{L}_i$ by adding  vectors $\bs{v}$ of the form in eq.~\eqref{eq:vGlue},  where each $\bs{v_i}\in \mathcal{L}_i^{\perp}/\mathcal{L}_i$.  We also refer to the set of extra vectors $\bs{v}$ as the \textit{glue group} $G$.
Let $\mathcal{L}= \bigoplus_{i=1}^k \mathcal{L}_i \cup G$ be a glued lattice of this form.  It is clear that we have $\mathcal{L}^{\perp} = \bigoplus_{i=1}^k \mathcal{L}_i^{\perp} \cap G^{\perp}$.

Generally,  whenever we have \textit{saturated} sublattices $\mathcal{L}_0\subset \mathcal{L}$,  i.e.,  $\mathbb{R}\otimes\mathcal{L}_0=\mathbb{R}\otimes\mathcal{L}$,  we have $|\mathcal{L}_0|/|\mathcal{L}|\in\mathbb{Z}$.  For $g\in \mathcal{L}$ we consider the glue classes $[g]=g+\mathcal{L}_0$,  which form an additive group that we denote 
\begin{equation}
G=\langle g_1,\dots, g_r\rangle.  
\end{equation}
Conversely,  in a bottom-up approach,  a glued lattice $\mathcal{L}$ can be constructed by considering a general glue group $G\subseteq \CL_0^{\perp}/\CL_0$ and forming $\mathcal{L}=\mathcal{L}_0[G]=\mathcal{L}_0 \cup G$.   $G$ is cyclic and isomorphic to $\mathbb{Z}_{n_1}\times ..\times \mathbb{Z}_{n_r}$ ,  where $\mathbb{Z}_n$ is the cyclic group of order $n$,  and each $n_i$ is the order of the corresponding generator $g_i$,  i.e  the smallest positive integer such that $n_i g_i \in \mathcal{L}_0$ .
The determinant of the glued lattice $\mathcal{L}_0[G]$ can be computed as  \cite{gannon_thesis} %
\begin{equation}
|\mathcal{L}_0[G]|=|\mathcal{L}_0|/|G|=|\mathcal{L}_0|/(\prod_{i=1}^r n_i). \label{det_glued}
\end{equation}%
To construct a \textit{symplectic} glued lattice from a symplectic base lattice $\mathcal{L}_0$,  it is important to take care that every $g_i$ has integer symplectic inner product with every other $g_j$ -- i.e., $G$ is itself a finite symplectic group and that each $G$ has integer symplectic inner product with each $x\in\mathcal{L}_0$,  i.e.,  $G\subset \mathcal{L}_0^{\perp}$.  It is easy to see that the earlier considerations are reproduced for $\mathcal{L}_0=\bigoplus_{i=1}^k \mathcal{L}_i$. 
Using \eqref{det_glued} we can obtain the logical  dimension of GKP codes associated to glued lattices $\mathcal{L}_0[G]$.  E.g.  for a concatenated GKP-qubit code,  $G\subset\mathcal{L}_0^{\perp}/\mathcal{L}_0$ is identified with the outer code with,  say,  $r$ lineary independent generators,  each with order $n_i=2$ in $\mathcal{L}_0$.  $\mathcal{L}_0=\mathcal{L}(\sqrt{2}I_{2n})$,  such that we compute $|\mathcal{L}_0[G]|=2^{n-r}=2^{k}$,  consistent with what we would expect.
Similar to the derivation for concatenated codes,  we do not expect it to be hard to derive the distance of a glued lattice $\mathcal{L}_0[G]$.  However, this quantity strongly depends on the specifications of the glue group $\CL_0^{\perp}/\CL_0$ and $G$ which makes it difficult to write down a general solution in closed form.
\je{Is it obvious to see when a glued lattice construction gives rise to codes that are not concatenated
codes?}\jc{e.g. consider e.g.  a qubit surface code where some of the qubits are replaced by higher level qudits - logical strings that go through these qudits now have different ways to pass through them (i.e.  we can choose e.g.  $X$ or $X^2$ to be the local tensor product factor).  These different ``pathes" have different lengths.}

\subsection{Tensor product codes}

Aside from the glueing construction,  it is also possible to obtain new codes by taking outer products of lattices.  The idea behind this construction is akin to product constructions known for qubit  quantum error correcting codes,  namely the  \emph{hypergraph product codes} by Tillich and Zemor \cite{Tillich_2014}, where the defining structure of the code is a hypergraph, and Homological product codes by Bravyi and Hastings \cite{HomologicalProduct}, where the code is defined via a cell complex.  
For GKP codes,  the defining structure of the codes is given by a lattices,  such that the tensor product for lattices serves as an immediate candidate for a similar construction.  

Let $\mathcal{L}_1=\mathcal{L}(M_1)\subset \mathbb{R}^2$ be a symplectic lattice with  symplectic Gram matrix $A_1=M_1J_2M_1^T$ and  $\mathcal{L}_2=\mathcal{L}(M_2)\subset \mathbb{R}^{n}$ an integral lattice with euclidean Gram matrix $G_2=M_2M_2^T$.  The tensor product lattice is defined as $\mathcal{L}_{\otimes}=\mathcal{L}(M_1\otimes M_2)=\mathcal{L}_1\otimes \mathcal{L}_2 \subset \mathbb{R}^{2n}$,  i.e.,  a basis for $\mathcal{L}_{\otimes}$ is given by $\lrc{(M_1)_i\otimes  (M_2)_j, \; i=1,2,\; j=1,\hdots,n}$. $\mathcal{L}_{\otimes}$ is a symplectic lattice due to the decomposition $J_{2n}=J_2\otimes I_n$, and its symplectic Gram matrix reads%
\begin{equation}
A_{\otimes}=(M_1\otimes M_2)J_{2n} (M_1\otimes M_2)^T=A_1 \otimes G_2,
\end{equation}%
which is integral by construction.  The canonical dual basis is given by
\begin{equation}
M^{\perp}_{\otimes}=(J_{2n}M_{\otimes}^T)^{-1}=M_{\otimes}^{-T}J_{2n}^T=M_1^{-T}J_2^T\otimes M_2^{-T}=M_1^{\perp}\otimes M_2^{*},
\end{equation}
which forms a basis for the symplectically dual lattice $\mathcal{L}_{\otimes}^{\perp}=\mathcal{L}_{1}^{\perp}\otimes \mathcal{L}_{2}^{*}$.
We have 
\begin{align}%
|A_1 \otimes G_2|&=|A_1|^n|G_2|^2,\\
k_{\otimes}=\frac{1}{2}\log_2 |A_1 \otimes G_2| &= \frac{n}{2}\log_2 |A_1| + \log_2 |G_2|.
\end{align}%
\begin{them}[Distance of tensor product codes]
The distance of the tensor product code $$\Delta_{\otimes}=\min_{0\neq x\in \mathcal{L}_{\otimes}^{\perp}/\mathcal{L}_{\otimes}} \|x\|$$ obeys %
\begin{equation}
\max\left\{ \frac{\Delta_1}{\lambda_n(\mathcal{L}_2)}, \frac{\Delta_2}{\lambda_2(\mathcal{L}_1)}\right\} \leq \Delta_{\otimes} \leq \Delta_1\Delta_2,
\end{equation}%
where %
\begin{equation}\Delta_1=\min_{0\neq x\in \mathcal{L}_1^{\perp}/\mathcal{L}_1} \|x\|,  \hspace{1cm} \Delta_2=\min_{0\neq x\in \mathcal{L}_2^{*}/\mathcal{L}_2} \|x\|.
\end{equation}%
\end{them} 

\begin{proof}
The proof is analogous to that of Lemma 2 in ref.~\cite{HomologicalProduct}.
To prove the upper bound,  let $x\in \mathcal{L}_{1}^{\perp}/\mathcal{L}_{1},\; y\in \mathcal{L}_{2}^{*}/\mathcal{L}_{2}$,  be minimal non-trivial  logical representatives of each component codes.  It is clear that $x \otimes y \in \mathcal{L}_{\otimes}^{\perp}$.  Further we have $x \otimes y \notin \mathcal{L}_{\otimes}$ because we can pick $a \in \mathcal{L}_1^{\perp},\; b \in \mathcal{L}_2^{*}$ such that $a^TJ_2x\in \mathbb{Q}_1,\, b^Ty\in \mathbb{Q}_1$,  where $\mathbb{Q}_1=\mathbb{Q}\cap (0,1)$,  yielding $(a\otimes b)^T J_{2n}(x\otimes y) \notin \mathbb{Z}$ (note that the symplectic inner product sets the commutation phase for the associated displacement operators).  As such,  we have  obtained a non-trivial logical operator $x\otimes y \in \mathcal{L}_{\otimes}^{\perp}/\mathcal{L}_{\otimes}$ and $\Delta_{\otimes}\leq \|x\otimes y\|=\Delta_1 \Delta_2$.
Let $0\neq\bs{\psi} \in \mathcal{L}_{\otimes}^{\perp}/\mathcal{L}_{\otimes}$ be a minimal length non-trivial vector.  We can always choose $0\neq c\in\mathcal{L}_1^{\perp}/\mathcal{L}_1$ and $d\in \mathcal{L}_2$ such that %
\begin{equation}(c\otimes d)^T (J_2\otimes I_n) \psi\notin\mathbb{Z}.\label{eq:proof_assumption_1}
\end{equation} %
Let  $\mathit{vec}^{-1}(\psi)$ be the un-vectorization of $\psi$,  i.e., if $\psi=\sum_i x_i \otimes y_i,  \, x_i\in \mathcal{L}_1^{\perp},\,y_i \in \mathcal{L}_2^{*}$  we have $\mathit{vec}^{-1}(\psi)=\sum_i x_i y_i^T$.  It holds that $\mathit{vec}^{-1}(\psi)d\in \mathcal{L}_1^{\perp}$ and $\mathit{vec}^{-1}(\psi)d \notin \mathcal{L}_1 $ since otherwise $(c\otimes d)^T (J_2\otimes I_n) \psi = c^TJ_2\mathit{vec}^{-1}(\psi) d \in\mathbb{Z}$ for all choices of $\bs{c},  \bs{d}$,  such that $\psi$ is logically trivial,  which is not the case by assumption.  Hence,  $\mathit{vec}^{-1}(\psi)d $ is a non-trivial representative of $\mathcal{L}_1^{\perp}/\mathcal{L}_1$.  Using Cauchy-Schwartz we have $\|\psi\|\|d\|=\|\mathit{vec}^{-1}(\psi)\|_F\|d\|\geq  \|\mathit{vec}^{-1}(\psi)d\|\geq \Delta_1$,  where $\|\cdot\|_F$ is the Frobenius norm.  
\je{Earlier, this work talks about the 2-norm. What is the Frobenius norm if not the
2-norm?}\jc{I'm just trying to consistently use $2-$norm for vector norms and Frobenius for something that takes a matrix as input.}
We can always choose $d\in\mathcal{L}_2$ with length at most the covering radius $\lambda_n(\mathcal{L}_2)$ that satisfies eq.~\eqref{eq:proof_assumption_1}.  This is because any $d\in\mathcal{L}_2$  can be written in a basis 
\begin{equation}
d=\sum_{i=1}^n a_i \xi_i, \hspace{.5cm} a_i\in \mathbb{Z},\hspace{.5cm}  \xi_i\in \CL_2 :\; \|\xi_i\|\leq \lambda_n\lr{\CL_2}.
\end{equation}
When $d$ satisfies eq.~\eqref{eq:proof_assumption_1},  we have
\begin{equation}
\mathbb{Z}\not\ni \sum_{i=1}^n (c\otimes d)^T (J_2\otimes I_n) \psi = \sum_{i=1}^n a_i (c\otimes \xi_i)^T (J_2\otimes I_n) \psi.
\end{equation}
There must be at least one summand $i=x$,  for which $a_x (c\otimes \xi_x)^T (J_2\otimes I_n) \psi \not\in \mathbb{Z}$.  Since $a_x \in \mathbb{Z}$,  it must hold that $(c\otimes \xi_x)^T (J_2\otimes I_n) \psi \not\in \mathbb{Z}$. 
Finally, we obtain 
\begin{equation}
\|\psi\|\geq \frac{\Delta_1}{\lambda_n(\mathcal{L}_2)}.
\end{equation}
Following the same procedure,  we choose $c\in \mathcal{L}_1,d\in \mathcal{L}_2^{\perp}/\mathcal{L}_2$ such that eq.~\eqref{eq:proof_assumption_1} is satisfied to show that $\|\psi\| \|c\|\geq \Delta_2$.
\je{The upper and lower bounds can be very different, but this is what it is, right? Without further
structure one cannot tighten them, I presume.}\jc{yes,  i dont think there is more to say...}
\end{proof}

\section{Conclusions and future work \label{sec:conclusions}}

In this work, we have shown how some of the core concepts of lattice theory can be fruitfully used to analyse the structure of GKP codes.  We have introduced the notion of a euclidean code distance for GKP codes and shown its relation with other relevant code parameters,  in particular the trade-off with the encoding rate and length of stabilizer operators.  We have examined the structure of particular GKP code constructions,  such as scaled and concatenated GKP codes and shown how extensions of those constructions can be obtained using tools already present in the lattice theory literature.   Using a lattice theoretic perspective when discussing bounds on GKP code parameters,  we observed that the distance of the code is intimately related to the length of the stabilizers.  Leveraging the natural embedding of qubit codes in $\mathbb{R}^{2n}$ provided by the concatenation with GKP codes we have shown in Corollary \ref{cor:qubit_distance_weights} how this result leads to a new proof of a similar property for qubit  stabilizer codes,  which is simpler than those found in the literature.  This highlights how concepts from lattice theory can also become valuable tools to understand qubit stabilizer codes.  Concerning the role of symplectic operations,  which are central both in the theory and experiments with bosonic systems, we have clarified which GKP codes are related through a symplectic transformation, extending previous results for the single-mode case.  This can also be interpreted as code-switching via Gaussian unitaries.  In particular, we have identified a necessary and sufficient condition in the equality of the symplectic Gram matrix in standard form.

We have discussed how the GKP codes highlighted in this work can be decoded in terms of a general maximum likelihood decoder for the GKP code,  for which the defining logical coset probability is given by the lattice theta function.  Algorithms to approximate or sample  theta functions have been discussed in the literature \cite{modularforms_approach},  which we identify as possible route forward to build general purpose decoders for GKP codes.  More details on the practical implementation of decoders are,  however,  to appear in a separate publication. 

The theta function of the GKP code lattice is also proven to encode information about the distance of GKP codes.  It is known that,  depending on the lattice,  the theta function may possess certain symmetries under conformal transformations,  which allow it to be represented in a finitely generated basis.  In other words,  if the theta function has certain symmetries (e.g.  invariance under $z\mapsto z+2$ for integral lattices) there exists a finite basis classified by the symmetry which can be used to write the theta function.  This can also potentially be useful for decoding.  In general,  it would be interesting to further investigate the connection between the code distance and possible conformal symmetries of the theta function of the lattice.  

It is worth mentioning that the theta function of glued lattices also appears in place of the partition function of certain string theories \cite{gannon_thesis}.  Given that it represents the MLD probability discussed in the Section \ref{sec:MLD},  which has also been interpreted in the past as  the partition function of a \textit{statistical mechanics model} \cite{Dennis_2002, toricGKP},  it is an interesting open question whether a rigorous connection between string theories and  the statistical mechanics model for decoding GKP codes can be leveraged to obtain a deeper understanding for either.

Finally,  we have discussed constructions of GKP codes beyond scaled- and concatenated codes.  We have seen that the toolbox of lattice theory allows for a remarkably versatile set of code constructions which have great potential for the search of good code families with potentially better parameters,  better performance tailored to particular physical implementations or allowing for better decoders relying on the structure of the lattice,  without reference to an intermediate qubit decoder.
Further generalizations of GKP code constructions,  such as subsystem variants by constructions on  \textit{shifted lattices} \cite{gannon_thesis} also seem possible.  Overall,  we hope that this work will stimulate further interest in examining stabilizer codes,  and in particular GKP codes,  through the lens of lattice theory and will provide a foundation to import useful tools for the study of bosonic error correction. 

\begin{acknowledgments}
J.~C.~thanks B.~Terhal for asking whether product constructions for GKP codes can be found and N.~P. ~Breuckmann,  A. ~Ciani and J.-P.~Seifert for helpful discussions.  F.~A.~acknowledges the support of the Alexander von Humboldt foundation.  This work has been funded by the BMBF (RealistiQ,  for which it provides a key physics 
non-agnostic approach to quantum error correction, PhoQuant, QPIC-1, and QSolid) and the DFG (CRC 183,  project B04, on entangled states of matter).  We also gratefully acknowledge support from the Einstein Research Unit on quantum devices.

\end{acknowledgments}

\appendix

\section{Code space of GKP codes}

In this section,  we show how the code space of a general GKP code can be represented in terms of the Wigner function of the code space projector
\begin{equation}
\Pi_{GKP}=\sum_{g\in \mathcal{S}} g.
\end{equation}
In particular,  this illustrates in how far the code space is non-trivially influenced by the choice of pivot basis for the GKP code and the phase-sector $\phi_M$.
Consistent with our non-standard definition of the displacement operator,  we define the generalized Wigner function of an operator $O$ \cite{Weedbrook_2012,  Continuous}
\begin{equation}
W_O\lr{\bs{x}} =\int_{\mathbb{R}^{2n}} d\bs{\eta}\, e^{-i2\pi \bs{x}^T J \bs{\eta}} \Tr\lrq{D\lr{\bs{\eta}} O},
\end{equation}
such that we also have
\begin{equation}
\int_{\mathbb{R}^{2n}} d\bs{x} W_O\lr{\bs{x}}=\Tr\lrq{O}.
\end{equation}
And we equip us with the following lemma.

\begin{lem}[Dirac comb representation]\label{lem:dirac_comb}
Let $\CL \subseteq \mathbb{R}^{2n}$ be a lattice with generator $M$ with symplectic dual $\CL^{\perp}$,  we have
\begin{equation}
\sum_{\bs{\xi} \in \CL} e^{i2\pi \bs{\xi}^TJ\bs{z}}= \frac{1}{\det\lr{\CL}}\sum_{\bs{\xi}^{\perp} \in \CL^{\perp}} \delta^{2n}\lr{\bs{z}-\bs{\xi}^{\perp}}.
\end{equation}
\end{lem}
\begin{proof}
We compute straightforwardly using one dimensional Poisson resummation $\sum_{n \in \mathbb{Z}} e^{i2\pi n x} = \sum_{k \in \mathbb{Z}} \delta\lr{x-k} $
\begin{align}
\sum_{\bs{\xi} \in \CL} e^{i2\pi \bs{\xi}^TJ\bs{z}}
&=\sum_{\bs{a} \in \mathbb{Z}^{2n}} e^{i2\pi \bs{a}^T MJ\bs{z}}  \\
&=\prod_{j}\sum_{a_j \in \mathbb{Z}} e^{i2\pi a_j \lr{MJ\bs{z}}_j} \\
&=\prod_{j}\sum_{b_j \in \mathbb{Z}} \delta\lr{\lr{MJ\bs{z}}_j-b_j} \\
&= \frac{1}{\det\lr{\CL}} \sum_{\bs{b} \in \mathbb{Z}^{2n}} \delta^{2n}\lr{\bs{z}-\lr{\bs{b}^TM^{\perp}}^T} \\
&= \frac{1}{\det\lr{\CL}}\sum_{\bs{\xi}^{\perp} \in \CL^{\perp}} \delta^{2n}\lr{\bs{z}-\bs{\xi}^{\perp}}.
\end{align}

\end{proof}

We compute the Wigner function of the code space projector of a GKP code defined with primal basis $M$
\begin{align}
W_{\Pi_{GKP}}\lr{\bs{x}}
&= \int_{\mathbb{R}^{2n}} d\bs{\eta}\, e^{-i2\pi \bs{x}^T J \bs{\eta}} \Tr\lrq{D\lr{\bs{\eta}} \Pi_{GKP}} \\
&=\sum_{\bs{\xi}\in \CL}  e^{i\phi_M\lr{\bs{\xi}}}  \int_{\mathbb{R}^{2n}} d\bs{\eta}\, e^{-i2\pi \bs{x}^T J \bs{\eta}} \Tr\lrq{D\lr{\bs{\eta}} D\lr{\bs{\xi}} }\\
&=\sum_{\bs{\xi}\in \CL}  e^{i\phi_M\lr{\bs{\xi}}} e^{i2\pi \bs{x}^T J \bs{\xi}}. \label{eq:Wigner_phi}
\end{align}
When $\phi_M=0 \mod 2\pi$,  such as when the GKP code is defined with symplectically even lattice,  we can compute the sum using Lemma \ref{lem:dirac_comb},
\begin{equation}
W_{\Pi_{GKP}}\lr{\bs{x}}= \frac{1}{\det\lr{\CL}}\sum_{\bs{\xi}^{\perp} \in \CL^{\perp}} \delta^{2n}\lr{\bs{x}-\bs{\xi}^{\perp}}, \label{eq:Wigner_delta}
\end{equation}
such that the Wigner function is precisely given by a generalized Dirac comb over the dual lattice $\CL^{\perp}$.  When $\phi_M \neq 0 \mod 2\pi$,  however,  this is no longer true.  The above expression will also be substantially different from a mere shift of the Dirac comb over the lattice since $\phi_M$ is generally a \emph{quadratic} function. 
 We can express it compactly using the general \emph{Riemann theta function} \cite{Wirtinger}
\begin{equation}
\vartheta\lr{\bs{u}\, \vert\, F}=\sum_{\bs{m}\in \mathbb{R}^{2n}} e^{i2\pi \lr{\frac{1}{2}\bs{m}^TF\bs{m}+\bs{m}^T\bs{u}}},
\end{equation}
as
\begin{equation}
W_{\Pi_{GKP}}\lr{\bs{x}}= \frac{1}{\det\lr{\CL}}\vartheta\lr{-MJ\bs{x}\, \vert \,A_{\lowertriangle}}.
\end{equation}
Note,  however,  that $A_{\lowertriangle}$ does not have positive definite imaginary part so that,  similar to the Dirac comb above,  this expression does not converge.  This is a general artefact of exact GKP states discussed here which are unphysical as presented and,  in practice,  will be endowed with particular regularizing functions depending on the concrete physical implementation.  When such regularization is present,  the corresponding sum can be expressed using the following generalization of Poisson resummation.

\begin{lem}[Poisson resummation]\label{lem:poisson}
Let $\CL \subseteq \mathbb{R}^{2n}$ be a lattice with symplectic dual $\CL^{\perp}$,  we have for all Schwartz functions $f:\, \mathbb{R}^{2n} \rightarrow \mathbb{C}$,

\begin{equation}
\sum_{\bs{\xi}\in \CL} f\lr{\bs{\xi}}=\frac{1}{\det\lr{\CL}} \sum_{\bs{\xi}^{\perp} \in \CL^{\perp}} \hat{f}\lr{\bs{\xi}^{\perp}},
\end{equation}
where
\begin{equation}
\hat{f}\lr{\bs{x}}=\int_{\mathbb{R}^{2n}}d\bs{y}f\lr{\bs{y}}e^{i2\pi \bs{y}^TJ\bs{x}}
\end{equation}
is the symplectic Fourier transform of $f$.
\end{lem}

\begin{proof}
We proof this fact following the analogous proof using the usual Fourier transform provided in \cite{theta_course}.  First define
\begin{equation}
F\lr{\bs{z}}=\sum_{\bs{\xi}\in \CL} f\lr{\bs{z}+\bs{\xi}}.
\end{equation}
This function is invariant under transformations $\bs{z}\mapsto \bs{z}+\CL$,  such that it descends to a function on $\mathbb{R}^{2n}/\CL$ and has Fourier expansion
\begin{equation}
F\lr{\bs{z}}=\sum_{\bs{y} \in \CL^{\perp}} \hat{F}\lr{\bs{y}}e^{i2\pi \bs{y}^TJ\bs{z}},
\end{equation}
where
\begin{equation}
\hat{F}\lr{\bs{y}}=\frac{1}{\det\lr{\CL}} \int_{\mathbb{R}^{2n}/\CL}d\bs{x}\, F\lr{\bs{x}}e^{i2\pi \bs{x}^TJ\bs{y}}.
\end{equation}
Using lemma \ref{lem:dirac_comb} the correctness of the Fourier expansion can be straightforwardly verified.
Let $\CP \sim \mathbb{R}^{2n}/\CL$ be a fundamental domain of $\CL$.  We compute for $\bs{y} \in \CL^{\perp}$ (such that $e^{i2\pi\bs{x}^TJ\bs{y}}$ is well defined on $\bs{x} \in \mathbb{R}^{2n}/\CL$)
\begin{align}
\det\lr{\CL}\hat{F}\lr{\bs{y}}
&=\sum_{\bs{\xi}\in \CL} \int_{\bs{x}\in \CP}d\bs{x}\,  f\lr{\bs{x}+\bs{\xi}}e^{i2\pi\bs{x}^TJ\bs{y}} \\
&=\sum_{\bs{\xi}\in \CL} \int_{\bs{x}\in \CP-\bs{\xi}}d\bs{x}\,  f\lr{\bs{x}}e^{i2\pi\bs{x}^TJ\bs{y}} \\
&=\hat{f}\lr{\bs{y}}, 
\end{align}
such that we have
\begin{equation}
F\lr{\bs{z}}=\frac{1}{\det\lr{\CL}} \sum_{\bs{\xi}^{\perp} \in \CL^{\perp}} \hat{f}\lr{\bs{y}} e^{i2\pi \bs{y}^TJ\bs{z}}.
\end{equation}
Taking $z=0$ completes the proof.
\end{proof}

We use Lemma \ref{lem:poisson} to compute the Wigner function of a GKP code,  where we include generic Gaussian regularization parametrized by $\epsilon > 0$.
\begin{align}
W_{\Pi_{GKP}}\lr{\bs{x}} 
&= \lr{\frac{\epsilon}{\pi}}^n \sum_{\bs{\xi}\in \CL}   e^{-\epsilon \bs{\xi}^T\bs{\xi} }  e^{i\phi_M\lr{\bs{\xi}}} e^{i2\pi \bs{x}^T J \bs{\xi}}\\
&=\lr{\frac{\epsilon}{\pi}}^n \sum_{\bs{\xi}\in \CL} e^{-\bs{\xi}^T \lr{\epsilon I - i M^{-1}A_{\lowertriangle}M^{-T}}\bs{\xi}+ i2\pi \bs{x}^TJ\bs{\xi}} \\
&=\frac{1}{\det\lr{\CL}}  \frac{\epsilon^n}{\sqrt{\det\lr{B}}} \sum_{\bs{\xi}^{\perp} \in \CL^{\perp}}  e^{-\pi^2 \lr{\bs{x}-\bs{\xi}^{\perp}}^T J B^{-1} J^T\lr{\bs{x}-\bs{\xi}^{\perp}} }.
\end{align}
where we have defined $B=\epsilon I - i M^{-1}A_{\lowertriangle}M^{-T}$.  We observe that for $A_{\lowertriangle}=0$ in the limit $\epsilon \rightarrow 0$ we again obtain eq.~\eqref{eq:Wigner_delta}.  However,  in general for odd $A_{\lowertriangle}$ the Wigner function is modified non-trivially.  In this example the regularization was chosen generically without reference to a clear physical motivation other than that physical states need to be normalizable.  It would be interesting to further study how the structure of a physical model of the code space changes with $A_{\lowertriangle}$ in the Wigner representation for alternative regularization schemes such as ones discussed in ref.  \cite{Mensen_2021}.

\section{GKP Codes from degenerate lattices}\label{appendix:degenerate}

In this section we briefly discuss how degenerate symplectic lattices can be used for encoding continuous degrees of freedom into a collection of $n$ oscillators,  following the definition and analysis in \cite{Noh_2020,  oscillator_to_oscillator}.  To simplify the notation we will work in the convention where $\bs{\hat{x}}=\lr{\hat{q}_1, \hat{p}_1, \hdots,  \hat{q}_n, \hat{p}_n,}$ and the symplectic form is $J=I_n \otimes J_2$.
We encode a state vector $\ket{\psi}\in\mathcal{H}^{\otimes k}$ describing $k$ oscillator modes into $\mathcal{H}^{\otimes n}$ by appending it with the ancillary state vector $\ket{A}=\ket{GKP}^{\otimes r}$,  where $r=n-k$,  and acting on the composite state with a Gaussian unitary $U_S$ (see fig.  \ref{fig:O2O}).
Each state vector $\ket{GKP}$ is the unique code state of the GKP stabilizer group described by $M_{GKP}=I_2$,  the symplectic self-dual square lattice,  such that the overall ancillary state vector $\ket{A}$ is stabilized by a GKP stabilizer group $\mathcal{S}_A$ associated to a lattice generated by $M_{A}=I_{2r}$ which we denote by $\CL_A=\CL\lr{M_A}$.
The code state is hence described by 
\begin{align}
\ket{\overline{\psi}} &=U_{S^{-1}}\lr{\ket{\psi} \otimes \ket{A}} \nonumber\\
&=U_{S^{-1}}\lr{\ket{\psi} \otimes D\lr{\bs{s}} \ket{A}} \nonumber \\ 
&=U_{S^{-1}} D\lr{\bs{0}_{2k}\oplus \bs{s}}U_{S^{-1}}^{\dagger}\ket{\overline{\psi}}  \nonumber \\ 
&= D\lr{S\lr{\bs{0}_{2k}\oplus \bs{s}}}\ket{\overline{\psi}} \hspace{.5cm} \forall \bs{s} \in \CL_A
\end{align}
and is stabilized by displacements associated to the lattice $\CL_{osc}$ generated by the rows of 
\begin{equation}
M_{osc}=\lrq{0_{2k} \oplus M_{A}}S^T,
\end{equation}
where we used $0_{2k}$ to denote the $2k \times 2k$ zero-matrix.  We denote the corresponding stabilizer group by $\mathcal{S}_{osc}$.
It is immediate that $M_{osc}$ is not a full rank matrix,  such that there will be a continuum of solutions $\bs{x}$ to 
\begin{equation}
MJ\bs{x} = 0 \mod 1.
\end{equation}
In particular,  we can obtain logical displacements by propagating  a logical displacement $D\lr{\bs{\alpha}}$ on the data mode through the encoding unitary
\begin{equation}
U_{S^{-1}}\lr{D\lr{\bs{\alpha}}\ket{\psi} \otimes \ket{A}} = D\lr{S\lr{\bs{\alpha}\oplus \bs{0}_{2r}}}\ket{\overline{\psi}}.
\end{equation}
Finally,  we obtain that the centralizer of the stabilizer group $\mathcal{S}_{osc}$ is given by displacements within
\begin{equation}
\CL_{osc}^{\perp}=\lrc{S\lr{\bs{\alpha} \oplus  \bs{0}_{2r}} + \bs{\xi}\big\vert\,  \bs{\alpha} \in \mathbb{R}^{2k},  \, \bs{\xi} \in \CL_{osc},  }.
\end{equation}
In contrast to the encoding of qubits and qudits,  where logical operations are encoded in the discrete collection of cosets of the quotient lattice,  the oscillator-to-oscillator encoding does not immediately inherit error correction properties from an underlying notion of (euclidean) distance,  where displacements of amplitude below a certain threshold are guaranteed to not harm the logical information.  Instead,  the protection stems from that fact that the encoding unitary can be tailored to correlate error displacements between the data- and the ancillary modes such that measuring the stabilizers $\mathcal{S}_A$ of the ancillary system allows to diagnose the error that incurred on the data modes.  That is,  when the error is promised to be sufficiently small,   the modular quadrature information gained from measuring $\mathcal{S}_A$ reveals the actual displacement error value.  Let $\mathcal{N}$ denote the Gaussian displacement channel.  We then have that $\rho_0=\ketbra{\psi} \otimes \ketbra{A}$ undergoes the evolution 

\begin{align}
\rho_0 
&\mapsto U_{S^{-1}}^{\dagger }  \int d^{2n} \bs{x}\,  e^{-\frac{1}{2} \bs{x}^T \lr{\frac{I}{\sigma^2}}\bs{x}} D\lr{\bs{x}} U_{S^{-1}} \rho_0 U_{S^{-1}}^{\dagger } D\lr{\bs{x}}^{\dagger}  \, U_{S^{-1}}\nonumber \\
&=   \int d^{2n} \bs{x}\,  e^{-\frac{1}{2} \bs{y}^T \lr{\frac{S^TS}{\sigma^2}}\bs{y}} D\lr{\bs{y}}  \rho_0 D\lr{\bs{y}}^{\dagger}, 
\end{align}
that is $\rho_0$ effectively undergoes a Gaussian displacement channel with covariance $\Sigma=\sigma^2 \lr{S^TS}^{-1}$.  This scheme can be understood as generalization and QEC application of the displacement-sensor setup developed in \cite{Duivenvoorden_Sensor}.  For concrete examples and further details on the oscillator-to-oscillator encoding we refer the reader to ref.  \cite{Noh_2020} and for analysis of the possibility of a threshold to ref.  \cite{oscillator_to_oscillator}.  

\begin{figure}[H]
\center
\includegraphics[width=.6\textwidth]{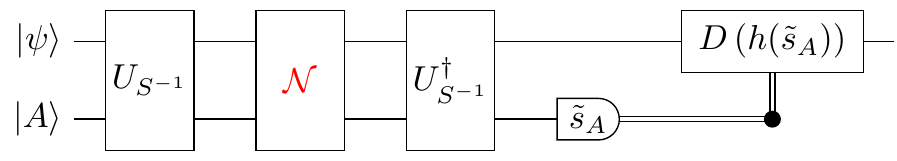}
\caption{Oscillator-to-oscillator encoding of ref.~\cite{Noh_2020}.  $\tilde{s}_A$ denotes the Eigenvalues obtained from measuring the generating set of stabilizers in $S_A$ and $h(\cdot)$ an estimator function that estimates the logical shift based on $\tilde{s}_A$.}\label{fig:O2O}
\end{figure}

To go beyond the encoding of either continuous or discrete information into a collection of oscillators,  we remark that this scheme can be extended to host both continuous and  discrete information.   One way to do so is by choosing as input ancillary state vector $\ket{A}$ a code state of a non-trivial GKP code described by a lattice $\CL\lr{M_A}$ with $|\det\lr{M_AJ_{2r}M_A^T} |>1$ such as one discussed in the main body of this work.  Here,  the ancillary state on the second register $\mathcal{H}^{\otimes r}$ would host all the discrete information with the addition that the errors applied to it may be (depending on the choice of $S$) strongly correlated with the errors on the first register $\mathcal{H}^{\otimes k}$ that hosts continuous information.  The stabilizer measurements $\mathcal{S}_A$ would hence also reveal information about the exact displacement the first register $\mathcal{H}^{\otimes k}$ has undergone where $\lambda_1\lr{\CL_A^{\perp}}$ sets the scale for the variance up to which the modular measurement outcomes $\tilde{s}_A$ still accurately estimate the exact shift on the first register (it needs to lie inside the dual Voronoi cell $\mathcal{V}^{\perp}$ of $\CL_A^{\perp}$).  As far as we know 
(see, e.g., Theorem 3) this does not need to limit the euclidean distance $\Delta$ with which the discrete information is protected.  We note that this hybrid scheme is equivalent to a setup where a GKP encoded discrete variable state is considered part of the CV input $\ket{\psi}$.  A more in-depth analysis of such hybrid scheme extends beyond the scope of this work and is therefore deferred to future work.

\bibliographystyle{quantum}
\bibliography{bib_lattice}

\end{document}